\documentclass[12pt]{article}
\usepackage{amsmath, amssymb, mathtools}
\usepackage{graphicx,psfrag,epsf}
\usepackage{chngcntr}
\usepackage[page]{appendix}

\usepackage{enumerate}
\usepackage{natbib}
\usepackage{url} % not crucial - just used below for the URL 

\newcommand{\blind}{1}

\usepackage{amsfonts, bm, amsthm, color}
\usepackage[ruled,vlined]{algorithm2e}
\usepackage{hyperref}
\usepackage{caption} 
\hypersetup{colorlinks=true, citecolor=blue, urlcolor=green}

\makeatletter
\def\@mb@citenamelist{cite,citep,citet,citealp,citealt,citepalias,citetalias}
\makeatother
\usepackage{multibib}
\newcites{supp}{References}

\theoremstyle{plain} \newtheorem{theorem}{Theorem} \newtheorem{proposition}{Proposition} \newtheorem{lemma}{Lemma} \newtheorem{corollary}{Corollary}

\theoremstyle{definition}   

\theoremstyle{remark}

\newcommand{\bigzero}{\mbox{\normalfont\Large\bfseries 0}}

\begin{document}

\def\spacingset#1{\renewcommand{\baselinestretch}%
{#1}\small\normalsize} \spacingset{1}

%%%%%%%%%%%%%%%%%%%%%%%%%%%%%%%%%%%%%%%%%%%%%%%%%%%%%%%%%%%%%%%%%%%%%%%%%%%%%%

\if1\blind
{
  \title{\bf Warped Dynamic Linear Models for Time Series of Counts}
  \author{Brian King and 
    Daniel R. Kowal\thanks{Brian King is PhD Graduate, Department of Statistics, Rice University, Houston, TX (\href{mailto:bking@rice.edu}{bking@rice.edu}).  Daniel R. Kowal is Dobelman Family Assistant Professor, Department of Statistics, Rice University, Houston, TX (\href{mailto:Daniel.Kowal@rice.edu}{Daniel.Kowal@rice.edu}).
    This material is based upon work supported by the National Science Foundation: a Graduate Research Fellowship under Grant No. 1842494 (King) and SES-2214726 (Kowal).}\hspace{.2cm} \\
    Department of Statistics, Rice University}
    \date{\vspace{-5ex}} \maketitle
} \fi

\if0\blind
{
  \bigskip
  \bigskip
  \bigskip
  \begin{center}
    {\LARGE\bf Warped Dynamic Linear Models for Time Series of Counts}
\end{center}
  \medskip
} \fi

\medskip
\begin{abstract}
Dynamic Linear Models (DLMs) are commonly employed for time series analysis due to their versatile structure, simple recursive updating, ability to handle missing data, and probabilistic forecasting. However, the options for count time series are limited: Gaussian DLMs require continuous data, while Poisson-based alternatives often lack sufficient modeling flexibility. We introduce a novel semiparametric methodology for count time series by \emph{warping} a Gaussian DLM. The warping function has two components: a (nonparametric) transformation operator that provides distributional flexibility and a rounding operator that ensures the correct support for the discrete data-generating process.  We develop conjugate inference for the warped DLM, which enables analytic and recursive updates for the state space filtering and smoothing distributions. We leverage these results to produce customized and efficient algorithms for inference and forecasting, including Monte Carlo simulation for offline analysis and an optimal particle filter for online inference. This framework unifies and extends a variety of discrete time series models and is valid for  natural counts,  rounded values, and multivariate observations. Simulation studies illustrate the excellent forecasting capabilities of the warped DLM. The proposed approach is applied to a multivariate time series of daily overdose counts and demonstrates both modeling and computational successes. 
\end{abstract}

\noindent%
\vfill

\newpage
\spacingset{1.45} % DON'T change the spacing!
\section{Introduction}
\label{sec:intro}
Count time series data inherit all the complexities of continuous time series data: the time-ordered observations may be multivariate, seasonal, dependent on exogenous variables, and exhibit a wide variety of autocorrelation structures. At the same time, count data often present uniquely challenging distributional features, including zero-inflation, over-/underdispersion, boundedness or censoring, and heaping. Additionally, discrete data require distinct strategies for probabilistic forecasting, uncertainty quantification, and evaluation.  As modern datasets commonly feature higher resolutions and lengthier time series, computational tools for both online and offline inference and forecasting are in demand. Fundamentally, the goals in count time series modeling are similar to those in the continuous setting: forecasting, trend filtering/smoothing, seasonal decomposition, and characterization of inter- and intra-series dependence, among others.  

In this paper, we develop methods, theory, and computing tools for a broad class of multivariate state space models that address each of these challenges and objectives. The core model is defined by   \textit{warping} a Gaussian Dynamic Linear Model (DLM; \citealp{WestHarrisonDLM}),  which we refer to as a \emph{warped DLM} (warpDLM):
\begin{align}
        \label{r-t}
        \bm y_t &= h \circ g^{-1}(\bm z_t) && \text{(warping)} \\
        %\label{transform}
%        g(\bm y_t^*) &= \bm z_t && \text{(transformation)} \\
        \label{latentdlm}
        \{\bm z_t\}_{t=1}^T &\sim \text{DLM} && \text{(see \eqref{dlm-obs} and \eqref{dlm-evol})} 
\end{align}
where $\bm y_t \in \mathbb{N}^n$ is the observed count data and $\bm z_t \in \mathbb{R}^n$ is continuous latent data. %Here, $\bm y_t^*$ serves primarily as a explanatory bridge between model components, and in practice is not necessary to compute. 
The warping has two components: a rounding operator $h: \mathcal{T} \rightarrow \mathbb{N}^n$, which ensures the correct support for the discrete data-generating process, and a monotone transformation function $g: \mathcal{T} \rightarrow \mathbb{R}^n$, which endows (possibly nonparametric) flexibility in the marginal distributions. The latent DLM enables straightforward embedding of familiar dynamic modeling structures such as local levels, seasonality, or covariates, along with a natural way of dealing with missing data. Note that this article focuses on modeling counts, but the framework is easily adaptable for general integer-valued  or rounded data. 

To isolate the rounding and transformation operations, the warping operation may be decomposed into $\bm y_t = h(\bm y_t^*)$ and $\bm z_t = g(\bm y_t^*)$ for a latent continuous variable $\bm  y_t^*$. We emphasize that the warpDLM is fundamentally distinct from the ``transformed DLM" strategy of (i) fitting a Gaussian DLM to (possibly transformed,  e.g., logarithmically) count data and (ii) rounding the resulting (continuous) forecasts. The ``transformed DLM" fails to account for the  discreteness of the data in the model-fitting process and introduces a critical mismatch  between the \emph{fitted model} and the \emph{data-generating process}. If the terminal rounding step is omitted, then the data-generating  process is not discrete; yet unless the rounding step is included within the model-fitting process---as in the warpDLM---then the fitted model fails to account for crucial features in the data.  In particular, the rounding operation is a nontrivial component of the model: within the warpDLM, $h$ provides the capability to model challenging discrete distributional features, such as zero-inflation, boundedness, or censoring (see Section~\ref{sec:warpedmethod}). Transformation-only models are known to be ineffective in many settings, such as low counts \citep{dontlogtransform2010}, yet  the warpDLM excels precisely in this case (see Section~\ref{subsec:sim}). 

A key contribution of this article is to develop conjugate inference for the warpDLM. In particular, we show that  the warpDLM likelihood is conjugate to the \emph{selection normal distribution} (e.g.,  \citealp{ArellanoValle2006}). Based on this result, we derive analytic and recursive updates for the warpDLM filtering and smoothing distributions. Crucially, we provide direct Monte Carlo simulators for these distributions---as well as the count-valued forecasting distribution---and construct an optimal particle filter for online inference. These models, derivations, and algorithms remain valid in the multivariate setting, and provide significant advancements over existing latent models for count time series (e.g., \citealp{latentgaussiancount_2021}).  To the best of our knowledge, these results  are unique for multivariate count time series models.%, and suggest that  the simple yet incoherent ``transformed DLM" strategy 

The warpDLM offers a unified framework for several discrete data models (dynamic as well as static) and incorporates strategies which have proven successful in other related methods. In the non-dynamic realm, \cite{Siegfried_Hothorn_2020} demonstrated the benefits of learned transformations for discrete data linear regression, while  \cite{kowal2021semiparametric} adopted a related transformation and rounding strategy to model heaped count data. Both took a frequentist approach to estimation. Among Bayesian methods, \cite{canale2011bayesian} and \cite{canale2013nonparametric} similarly applied rounded Gaussian and Dirichlet processes, respectively, without the transformation considerations. \cite{Kowal2020a} showcased the advantage of both the rounding \emph{and} transformation components within a static regression setting, and \cite{kowalSTARconjugate} extended this framework to incorporate Bayesian nonparametric estimation of the transformation. Looking to dynamic models, when $g$ is viewed as a copula, the warpDLM resembles the count time series model proposed by \cite{latentgaussiancount_2021}. Unlike  \cite{latentgaussiancount_2021}, we do not focus on stationary latent Gaussian processes, but rather incorporate DLMs to enable nonstationary modeling, dynamic covariates, and Bayesian inference and forecasting within a familiar setting. The warpDLM also generalizes recent work done in the binary data space, in particular the dynamic probit model of \cite{Fasano2021}. In doing so, we construct novel theory utilizing a broader class of distributions and requiring distinct computations of relevant posterior quantities.

 %\cite{latentgaussiancount_2021} 

The warpDLM framework falls in the category of generalized state space models, which are traditionally separated into two classes, stemming from \cite{Cox1981}: parameter-driven and observation-driven. In observation-driven models, the state process is treated as an explicit function of past data values, and methods often utilize likelihood-based inference born out of GLM theory (see \cite{fokianos2015GLM} for a review). Bayesian methods, including the warpDLM framework, are more commonly parameter-driven, which means the latent state parameter is treated as stochastic; model learning typically proceeds via recursive updating, e.g., using the Kalman filter \citep{kalmanOG}. The DLM is the most well-known model in this category, but relies on Gaussian assumptions unmet by count data. Dynamic Generalized Linear Models (DGLMs) were developed to adapt the state space framework for non-Gaussian data within the exponential family \citep{West_Harrison_Migon_1985}. For most count data, Poisson is the only available observational density that belongs to the exponential family. Binomial DGLMs are applicable for bounded data, yet perform quite poorly in the case of zero-inflation or heaping on  the boundary  (see Section~\ref{subsec:sim}). The negative binomial distribution with fixed dispersion parameter is also exponential family, but there is often little guidance for determining the  dispersion parameter.  \cite{Berry_West_2019} recently extended the DGLM family by mixing Bernoulli and Poisson DGLMs to better model zero-inflated and overdispersed count data.

%, Furthermore, the authors design a latent factor approach to multivariate modeling in the sales context, although this approach requires a baseline model for the latent factors, which may be difficult to specify in other applications.

The Poisson DGLM and extensions provide ``data coherent'' \citep{Freeland_McCabe_2004} inferences and forecasts, in the sense that predictions are appropriate to the type of data (count time series),  but the closed-form Kalman filtering results are generally unavailable. In most DGLM specifications, the evolution equation is assumed Gaussian (see \eqref{dlm-evol}), which necessitates linearized approximations or MCMC algorithms  for smoothing, filtering, and forecasting. \cite{West_Harrison_Migon_1985} prioritized closed-form and conjugate recursions, yet require approximate and moment-based (rather than distributional) updates for state parameters. \cite{Fahrmeir1992} designed an extended Kalman Filter to estimate posterior modes in a multivariate setting.  \cite{durbinkoopman2000} used importance sampling based on a linear approximation. \cite{fruhwirth2006auxiliary} constructed an approximate Gibbs sampler using data augmentation and mixture sampling. Another strategy is to replace the Gaussian evolution equation.   \cite{Gamerman2013} used a Poisson model with a multiplicative state update to preserve analytic and recursive inference. %, generalizing the beta temporal evolution of \cite{Smith_Miller_1986}. 
However, this updating structure has limited dynamic flexibility, for example to include seasonality or covariates.  \cite{Aktekin2018} proposed a multivariate  extension of this multiplicative model, but the analytic updating results were not preserved and the multivariate structure only accommodated positive correlations among the series. Another exception is \cite{Bradley_Holan_Wikle_2018}, who proposed log-Gamma processes that are conditionally conjugate to the Poisson distribution. This model still requires Gibbs sampling for all state smoothing, filtering, and forecasting distributions.

%The state space approach can also be used in a hierarchical fashion to model complex data processes \citep{gamerman1993dynamic, Wikle_2002}. Bayesian hierarchical approaches for count data (often multivariate spatiotemporal data) typically assume a Poisson data model conditional on a latent Gaussian process \citep{Cressie2011}, but computation in such a model is often inefficient and slow, especially when the dimension is high. To address this, \cite{Bradley_Holan_Wikle_2018} develop conjugate inference for a Poisson model with latent log-Gamma process, such that full conditionals can be derived for a fast Gibbs sampler. Such a model was recently expanded to data models from any exponential family distribution in \cite{Bradley_Holan_Wikle_2019}. These frameworks define very broad models for complex, high-dimensional data, but the count data model itself is still mostly limited to the Poisson distribution, and thus may not be able to handle complexities like zero-inflation. Many other count time series models outside of Bayesian state space modeling have been developed, but a full discussion is beyond the scope of this article; readers can consult the text \cite{davis2016handbook} or the recent review by \cite{DavisCountTSReview2021}.

The common limitations among existing state space models for count data are (i) a lack of exact, coherent, and recursive updates for filtering, smoothing, and forecasting distributions and (ii) restricted options for count-valued distributions. The warpDLM framework directly addresses and overcomes both limitations.

The paper is organized as follows.  Section \ref{sec:warpedmethod} introduces DLMs, the proposed model, and examples for the rounding and transformation. In Section \ref{sec:theory}, we derive the smoothing, filtering, and forecasting distributions for the warpDLM.  We discuss computing strategies for online and offline inference in Section \ref{sec:comp}.  Finally, we present forecasting results on simulated data as well as a real-data application in Section \ref{sec:application} before concluding. Supplementary material includes proofs of presented theorems, additional simulations, and further details on the application data and model specification.  Code to reproduce all findings is also available on \href{https://github.com/bking124/warpDLM-reproducible-code}{GitHub}. warpDLM functionality is included in the \verb|R| package \verb|countSTAR| \citep{countSTAR}, which is available on \href{https://CRAN.R-project.org/package=countSTAR}{CRAN} and documented in an \href{https://bking124.github.io/countSTAR/articles/countSTAR.html}{online vignette}.

\section{Dynamic Linear Models and Time Series of Counts}
\label{sec:warpedmethod}
The broad success of Bayesian time series analysis has largely been driven by Dynamic Linear Models (DLMs), also known as linear state space models.  The DLM framework subsumes ARIMA models and provides decomposition of time series, dynamic regression analysis, and multivariate modeling capabilities.  Additionally, the sequential updating structure provides a simple way to deal with missing observations (cf. \citealp[Section~2.7]{Durbin_Koopman_2012}). DLMs are widely popular not only because of their versatility, but also because the filtering, smoothing, and forecasting distributions are available in closed-form via the recursive Kalman filter. 

A DLM is defined by two equations: (i) the observation equation, which specifies how the observations are related to the latent state vector and (ii) the state evolution equation, which describes how the states are updated in a Markovian fashion. We present the Gaussian DLM for an \emph{observable} continuous $n$-dimensional time series $\{\bm z_t\}_{t=1}^T$, but note that these continuous variables are latent (i.e., unobservable) within the warpDLM:
\begin{align}
    \label{dlm-obs}
    \bm z_t &= \bm F_t \bm \theta_t + \bm v_t, \quad \bm v_t \sim N_n(\bm 0, \bm V_t) \\
    \label{dlm-evol}
    \bm \theta_t &= \bm G_t \bm \theta_{t-1} + \bm w_t, \quad \bm w_t \sim N_p(\bm 0, \bm W_t)
\end{align}
for $t=1,\ldots, T$, where $\{\bm v_t, \bm w_t\}_{t=1}^T$ are mutually independent and $\bm \theta_0 \sim N_p(\bm a_0, \bm R_0)$. Depending on the DLM specification, the $p$-dimensional  state vector  $\bm \theta_t$ could describe a local level, regression coefficients, or seasonal components, among other features. Common choices of the $n \times p$ observation matrix $F_t$ include the identity $(n=p)$, a matrix  of indicators that selects certain elements of $\bm \theta_t$ (e.g., for structural time series models), and covariate values (e.g., for dynamic regression analysis). The  $p \times p$ state evolution matrix $\bm G_t$  is frequently the identity but can be more complex, for example to capture seasonality. The observation and evolution covariance matrices  are $\bm V_t$ ($n \times n$) and $\bm W_t$ ($p \times p$), respectively.     Taken together, the quadruple $\{\bm F_t, \bm G_t, \bm V_t, \bm W_t\}_{t=1}^T$  defines the DLM.  Often, these matrices will be time-invariant.

The Gaussian DLM \eqref{dlm-obs}--\eqref{dlm-evol} is data-incoherent for discrete time series: the forecasting distribution does not match the support  of the data. DGLMs offer one resolution by replacing the Gaussian observation equation with an exponential family distribution. However, the primary option for count  data is the Poisson distribution, which  is often inadequate and requires additional modeling layers for common discrete data features such  as zero-inflation, over/underdispersion, boundedness or censoring, and heaping. These additional layers introduce significant modeling and computational complexity. By comparisons, the warpDLM  is capable of modeling \emph{each} of these distributional features under default specifications via the warping (rounding and transformation) operation, yet  maintains the useful and familiar state space formulation through the latent DLM.  By leveraging Gaussian state space models, the warpDLM builds on the long history of theoretical and computational tools  \citep{WestHarrisonDLM, petrisDLM, Prado_Ferreira_West_2021}, and operates within a familiar setting  for practitioners.  

The warpDLM framework links count data $\bm y_t$ with a (latent) Gaussian DLM for $\bm z_t$ in \eqref{dlm-obs}--\eqref{dlm-evol} via the warping operation \eqref{r-t}.  The rounding operation $h: \mathcal{T} \rightarrow \mathbb{N}^n$ serves as the connection mechanism between the real-valued latent space and the non-negative integers, setting $h(\bm y^*)=\bm j$ for any $\bm y^* \in \mathcal{A}_j$. These pre-image sets $\mathcal{A}_j$ form a disjoint partition of the space $\mathcal{T}$. In the univariate setting, $\mathcal{A}_j = [a_j, a_{j+1})$ is simply an interval. The warpDLM likelihood can thus be written as
\begin{equation}
\label{likelihood-time}
\mathbb{P}(\bm y_t = j | \bm \theta) = %\mathbb{P}\{\bm y_t^* \in \mathcal{A}_j | \bm \theta\} = 
\mathbb{P}\{\bm z_t \in g(\mathcal{A}_j) | \bm \theta\}, \quad t=1,\ldots,T
\end{equation}
for $j \in \mathbb{N}^n$, where $\bm z_t \in g(\mathcal{A}_j)$  is defined elementwise when $\bm y_t$ is multivariate. %As noted previously, $\bm y_t^*$ helps clarify the model, but is unnecessary in implementation: the rounding and transformation steps can simply be chained into one warping operation $\bm y_t = h\left\{g^{-1}(\bm z_t)\right\}$.

Both components of the warping operation serve important purposes that lead directly to desirable model properties.  The rounding function ensures that the warpDLM has the correct support for the (possibly bounded or censored) count data.  For simplicity, suppose $n=1$; generalizations occur by applying these specifications elementwise. 
By default, we take the rounding function to be the floor function, so $\mathcal{A}_j = [j, j+1)$. In addition, we include the zero modification $g(\mathcal{A}_0) = (-\infty, 0)$ so that $ y_t = 0$ whenever $ z_t < 0$. This specification maps much of the latent space to zero, with persistence of zeros determined by the  DLM \eqref{dlm-obs}--\eqref{dlm-evol}, for example,  $\mathbb{P}( y_t = 0 |  y_{t-1} = 0) = \mathbb{P}( z_t < 0 |  z_{t-1} < 0)$.
Similarly, when there is a known upper bound $y_{max}$ due to natural bounds or censoring, we may simply set $\mathcal{A}_{y_{max}} = [y_{\max}, \infty)$ so that the warpDLM has the correct  support, $\mathbb{P}( y_t \le y_{max} | \bm\theta) = 1$.  These useful rounding operation properties are formalized in \cite{Kowal2020a} and \cite{kowal2021semiparametric}. Importantly, the rounding operator does not require any modification to the computing algorithms: once it is specified, inference proceeds the exact same way for all choices of $h$.

While the rounding operation matches the discreteness and support of  the data, the transformation enables (nonparametric) distributional flexibility.  We apply the transformation elementwise, $g(\bm y_t^*) = (g_1(y_{1,t}^*),\ldots, g_n(y_{n,t}^*))'$, and again present the $n=1$ case for simplicity. 
The only requirement of the transformation $g$ is that it be strictly monotonic, which preserves ordering in the latent data space and ensures an inverse exists. We present both \emph{parametric} and \emph{nonparametric} modeling strategies. Parametric examples include classical transformations of count data, such as logarithmic, square-root, or identity transformations, and introduce no additional parameters into the model. %beyond the Gaussian DLM \eqref{dlm-obs}--\eqref{dlm-evol}. 

The nonparametric strategy uses a flexible and data-driven approach to infer the transformation based on the marginal distribution of each component of $\bm y$. Specifically, let $\mathcal{A}_j = [a_j, a_{j+1})$ as above and consider $n=1$. The cumulative distribution function (CDF) of $ y$ and $ z$ are linked via $F_y(j) = F_z\{g(a_{j+1})\},$ which suggests the transformation
\begin{equation}\label{g-hat}
\hat g_0(a_{j+1}) = \bar y + \hat s_y \Phi^{-1}\{\tilde F_y(j)\},
\end{equation}
where $\tilde F_y$ is an estimate or model for $F_y$ and $\bar y$ and $\hat s_y$ are the sample mean and sample standard deviation, respectively, of $\{ y_t\}_{t=1}^T$, to  match the marginal moments of $y$ and $z$ \citep{kowal2021semiparametric}. We smoothly interpolate $( y_t, \hat g_0(a_{ y_t + 1}))$  using a monotonic spline, which ensures that the warpDLM is supported on $\mathbb{N}$ instead of only the observed data values. Hence, a (nonparametric) model for $g$ may be equivalently specified by a (nonparametric) model for $F_y$. We adopt the (rescaled) empirical CDF $\tilde F_y(j) = (T+1)^{-1} \sum_{t=1}^T\mathbb{I}(y_t  \le j)$, which implies a \emph{semiparametric} model for the warpDLM. Many other models for $F_y$ are compatible within the warpDLM, including Bayesian nonparametric  models and parametric distributions (e.g., Poisson or Negative  Binomial marginals). By design, the model for the transformation $g$ is decoupled from both the rounding operator $h$ and the DLM \eqref{dlm-obs}--\eqref{dlm-evol}, so the subsequent derivations  and algorithms require only trivial modifications for distinct  choices of $g$.

The nonparametric transformation serves as a reasonable default across a variety  of scenarios, as suggested by our results in Section \ref{sec:application} and further confirmed in other (non-time series) settings (e.g. \citealp{kowal2021semiparametric}). To compare models with different transformations, there are a variety of choices. In the simulation study of Section \ref{subsec:sim}, we select the best model based on out-of-sample forecasting performance using leave-future-out cross-validation. However, in practice, any Bayesian model comparison metric could be used, such as the widely applicable or Watanabe-Akaike information criterion (WAIC, \citealp{watanabe2010asymptotic}) or approaches more specifically tailored to time series analysis \citep{burkner2020approximate}. 

\section{Exact Filtering and Smoothing}
\label{sec:theory}
In this section, we derive the recursive updates for the warpDLM filtering, smoothing, and forecasting distributions. Currently, these \emph{exact} results stemming from a \emph{coherent} joint distribution \eqref{r-t}--\eqref{latentdlm} are unique among state space models for multivariate count (or rounded) data, and crucially enable MCMC-free inference and forecasting. 

\subsection{Selection Distributions and the warpDLM}
\label{subsec:selectionbasics}
Consider the first time step $t=1$ of the warpDLM. Here, we omit the time subscripts for simplicity. The latent data \eqref{latentdlm} are described by the two DLM equations \eqref{dlm-obs}--\eqref{dlm-evol}, which can be rewritten as a single equation
\begin{equation}
\label{eq:firststepsimple}
\bm z = \bm F \bm \theta + \bm v, \quad \bm v \sim N_n(\bm 0, \bm V) 
\end{equation}
 with $\bm \theta \coloneqq \bm \theta_1$ and the  prior $\bm \theta \sim N_p(\bm \mu_\theta= \bm G a_0, \bm \Sigma_\theta=\bm G \bm R_0 \bm G' + \bm W)$. Using the likelihood \eqref{likelihood-time}, the posterior distribution is
\begin{equation}
    \label{posterior-gen}
    p(\bm \theta | \bm y) = p(\bm \theta | \bm z\in \mathcal{C}) = \frac{p(\bm \theta) p\{\bm z \in \mathcal{C} | \bm \theta\}}{p\{\bm z \in \mathcal{C}\}}
\end{equation}
which is also the first-step filtering distribution. Within the warpDLM, conditioning on $\bm y$ is equivalent to conditioning on  $\bm z$ belonging to a set $\mathcal{C}= g(\mathcal{A}_{\bm y})$. Although the representation is general, this set is typically simple: the default floor operator for $h$ implies that  $\mathcal{A}_{\bm  y} = [y_1,  y_1 + 1) \times \cdots \times [y_n, y_n + 1)$, so that  $\bm  z \in \mathcal{C}$ implies that each element of $\bm z$ belongs to a (transformed) interval.

A distribution of the form \eqref{posterior-gen} is known as a selection distribution \citep{ArellanoValle2006}.  When the two variables $\bm \theta$ and $\bm z$ are jointly normal---as in the warpDLM---the resulting distribution is a \emph{selection normal} (SLCT-N). More formally, given the joint distribution 
\begin{equation*}
    %\label{slct-n}
    \begin{pmatrix}
    \bm z \\ \bm \theta 
    \end{pmatrix} \sim 
   N_{n + p} \left\{
    \begin{pmatrix} \bm \mu_z \\ \bm \mu_\theta \end{pmatrix},
     \begin{pmatrix} \bm \Sigma_z & \bm\Sigma_{z\theta} \\ \bm \Sigma_{z\theta}' & \bm  \Sigma_\theta \end{pmatrix}
    \right\}
\end{equation*}
we denote the conditional random variable $[\bm \theta | \bm z \in \mathcal{C}] \sim \mbox{SLCT-N}_{n, p}(\bm \mu_z, \bm \mu_\theta, \bm \Sigma_z, \bm \Sigma_\theta, \bm \Sigma_{z\theta},  \mathcal{C})$ for constraint region $\mathcal{C}$. This random variable has density 
\begin{equation}
    \label{density-slct-n}
    p(\bm \theta | \bm z \in \mathcal{C}) = \phi_p(\bm \theta; \bm \mu_\theta, \bm \Sigma_\theta) \frac{\bar \Phi_n(\mathcal{C}; \bm\Sigma_{z\theta} \bm \Sigma_\theta^{-1}(\bm \theta - \bm \mu_\theta) + \bm \mu_z, \bm \Sigma_z - \bm\Sigma_{z\theta}\bm \Sigma_\theta^{-1}\bm\Sigma_{z\theta}')}{ \bar\Phi_n(\mathcal{C}; \bm \mu_z, \bm \Sigma_z)}
\end{equation}
where $\phi_p(\cdot; \bm \mu, \bm \Sigma)$ denotes the Gaussian density function of a Gaussian random variable with mean $\bm \mu$ and covariance $\bm \Sigma$ and $\bar\Phi_n(\mathcal{C}; \bm \mu, \bm \Sigma) = \int_\mathcal{C} \phi_n(\bm x; \bm \mu, \bm \Sigma) d \bm x$.  The density \eqref{density-slct-n} is somewhat unwieldy in practice, but there is also a constructive representation which allows for direct Monte Carlo simulation from the posterior density $p(\bm \theta | \bm y)$ (see Section~\ref{sec:comp}).

For the first-step model \eqref{eq:firststepsimple} with prior $\bm \theta \sim N_p(\bm \mu_\theta, \bm \Sigma_\theta)$, we report the exact posterior distribution:

\begin{theorem}
    \label{normal-conjugacy}
    Under \eqref{eq:firststepsimple}, the posterior distribution is $[\bm \theta | \bm y] \sim \mbox{SLCT-N}_{n, p}(\bm \mu_z = \bm F \bm \mu_\theta, \bm \mu_\theta, \bm \Sigma_z = \bm F \bm \Sigma_\theta \bm F' + \bm V, \bm \Sigma_\theta, \bm \Sigma_{z\theta} = \bm F \bm \Sigma_\theta,  \mathcal{C} = g(\mathcal{A}_{\bm y}))$. 
\end{theorem}

The Gaussian prior in Theorem \ref{normal-conjugacy} is conjugate: Gaussian distributions can be viewed as a limiting case of a $\mbox{SLCT-N}_{n, p}$ where $n=0$. However, the conjugacy suggests a more general prior for $\bm \theta$, namely, the SLCT-N distribution:

%may be expressed as $[\bm \theta] \sim \mbox{SLCT-N}_{n_0, p}(\mu_z = 0, \bm \mu_\theta, \Sigma_z = 1, \bm \Sigma_\theta, \bm \Sigma_{z\theta} = \bm 0_p',  \mathcal{C} = \mathbb{R})$. The moments and constraints on $\bm z$ are irrelevant as long as $\bm \Sigma_{z\theta} = \bm 0$. 
\begin{lemma}
\label{lem:conjugacy}
    Consider the prior $\bm \theta \sim \mbox{SLCT-N}_{n_0, p}(\bm \mu_{z_0}, \bm \mu_\theta, \bm \Sigma_{z_0}, \bm \Sigma_\theta, \bm \Sigma_{z_0\theta},  \mathcal{C}_0)$ with the latent-data observation equation \eqref{eq:firststepsimple}. The posterior is 
    \begin{align*}
        [\bm \theta | \bm y] \sim \mbox{SLCT-N}_{n_0 + n, p}\Big\{
        &\bm \mu_{z_1} = \begin{pmatrix} \bm \mu_{z_0} \\ \bm F \bm \mu_\theta \end{pmatrix}, 
        \bm \mu_\theta, 
        \bm \Sigma_{z_1} = \begin{pmatrix} \bm \Sigma_{z_0} & \bm \Sigma_{z_0\theta} \bm F' \\ \bm F \bm \Sigma_{z_0\theta}' & \bm F \bm \Sigma_\theta \bm F' + \bm V \end{pmatrix},  
        \bm \Sigma_\theta, \\
        & \bm \Sigma_{z_1\theta} = \begin{pmatrix} \bm \Sigma_{z_0\theta} \\ \bm F \bm \Sigma_\theta \end{pmatrix},  
        \mathcal{C}_1 = \mathcal{C}_0 \times g(\mathcal{A}_{\bm y}) \Big\}.
    \end{align*}
\end{lemma}

Thus the selection normal distribution is conjugate with the warpDLM likelihood.  Moreover, since the Gaussian distribution is closed under linear transformations, this property is preserved for the selection normal distribution \citep{ArellanoValle2006}. These results link the selection normal distribution with the warpDLM prior-to-posterior updating mechanism, and provide the key building blocks for deriving the warpDLM filtering and smoothing distributions.
%Moreover, we can also show that the selection normal is closed under linear transformations. This fact will be vital in showing that the state predictive distribution stays within the selection normal family.

\subsection{Filtering and Smoothing for Warped DLMs}
\label{subsec:warpeddlmtheory}
We now proceed to the most general setting for warpDLM, using the conjugacy and closure under linear transformation results to derive the exact forms of the filtering, smoothing, and forecasting distributions. Throughout, we assume the quadruple $\{\bm F_t, \bm G_t, \bm V_t, \bm W_t\}$ for $t=1,\ldots,T$ is known. Estimation of variance parameters is addressed in  Section~\ref{sec:comp}. 

\subsubsection{Filtering}
Adding back the appropriate subscripts to Theorem \ref{normal-conjugacy},  the first-step filtering distribution is
%\begin{equation}
%\begin{aligned}
 $   (\bm \theta_1 | \bm y_1 ) \sim \mbox{SLCT-N}_{n,p} (\bm \mu_{z_1} = \bm F_1 \bm a_1,  \bm \mu_{\theta}=\bm a_1, \bm \Sigma_{z_1} = \bm F_1 \bm R_1 \bm F_1' + \bm V_1, 
    \bm \Sigma_\theta = \bm R_1, \bm \Sigma_{z_1\theta} = \bm F_1 \bm R_1, \bm C_1 = g(\mathcal{A}_{\bm y_1})),$
%\end{aligned}
%\end{equation}
where $\bm a_1=\bm G_1 \bm a_0$ and $\bm R_1=\bm G_1 \bm R_0 \bm G_1' + \bm W_1$. Given the first-step filtering distribution, we proceed inductively for time $t$:
\begin{theorem}
\label{theorem:filtering}
    Let $        (\bm \theta_{t-1} | \bm y_{1:t-1}) \sim \mbox{SLCT-N}_{n(t-1),p} (\bm \mu_{z_{t-1|t-1}}, \bm \mu_{\theta_{t-1|t-1}},  \bm \Sigma_{z_{t-1|t-1}}, 
         \bm  \Sigma_{\theta_{t-1|t-1}},  
         \\ \bm  \Sigma_{(z \theta)_{t-1|t-1}},  C_{t-1|t-1})
$
   % \begin{equation*}
      %  (\bm \theta_{t-1} | \bm y_{1:t-1}) \sim \mbox{SLCT-N}_{n(t-1),p} (\bm \mu_{z_{t-1|t-1}}, \bm \mu_{\theta_{t-1|t-1}},  \bm \Sigma_{z_{t-1|t-1}}, 
 %       \bm  \Sigma_{\theta_{t-1|t-1}},  \bm  \Sigma_{(z \theta)_{t-1|t-1}},  C_{t-1|t-1})
 %   \end{equation*}
   be the filtering distribution at time $t-1$ under the warpDLM. Then the one-step-ahead state predictive distribution at $t$ is 
   \begin{equation}
    \label{statepred}    
        (\bm \theta_{t} | \bm y_{1:t-1}) \sim \mbox{SLCT-N}_{n(t-1),p} (\bm \mu_{z_{t|t-1}}, \bm \mu_{\theta_{t|t-1}},  \bm \Sigma_{z_{t|t-1}}, 
        \bm  \Sigma_{\theta_{t|t-1}},  \bm  \Sigma_{(z \theta)_{t|t-1}},  C_{t|t-1})
    \end{equation}
    with  $\bm \mu_{z_{t|t-1}} = \bm \mu_{z_{t-1|t-1}}$, $\bm \mu_{\theta_{t|t-1}} = \bm G_t \bm \mu_{\theta_{t-1|t-1}}$, $\bm \Sigma_{z_{t|t-1}}= \bm \Sigma_{z_{t-1|t-1}}$, $\bm \Sigma_{\theta_{t|t-1}}= \bm G_t\bm \Sigma_{\theta_{t-1|t-1}}\bm G_t' + \bm W_t$, $\bm \Sigma_{(z\theta)_{t|t-1}}= \bm \Sigma_{(z\theta)_{t-1|t-1}}\bm G_t'$, and $C_{t|t-1}=C_{t-1|t-1}$.
    Furthermore, the filtering distribution at time $t$ is 
    \begin{equation}
    \label{filtering}    
        (\bm \theta_{t} | \bm y_{1:t}) \sim \mbox{SLCT-N}_{n(t),p} (\bm \mu_{z_{t|t}}, \bm \mu_{\theta_{t|t}},  \bm \Sigma_{z_{t|t}}, 
        \bm  \Sigma_{\theta_{t|t}},  \bm  \Sigma_{(z \theta)_{t|t}},  C_{t|t})
    \end{equation}
    with  $\bm \mu_{z_{t|t}} = \begin{pmatrix} \bm \mu_{z_{t|t-1}} \\ \bm F_t \bm \mu_{\theta_{t|t-1}} \end{pmatrix}$, 
    $\bm \mu_{\theta_{t|t}} = \bm \mu_{\theta_{t|t-1}}$, 
    $\bm \Sigma_{z_{t|t}}= \begin{pmatrix} \bm \Sigma_{z_{t|t-1}} & \bm \Sigma_{(z\theta)_{t|t-1}}\bm F_t' \\  \bm F_t \bm \Sigma_{(z\theta)_{t|t-1}}' & \bm F_t \bm \Sigma_{\theta_{t|t-1}}\bm F_t' + \bm V_t \end{pmatrix}$, 
    $\bm \Sigma_{\theta_{t|t}}= \bm \Sigma_{\theta_{t|t-1}}$, 
    $\bm \Sigma_{(z\theta)_{t|t}}= \begin{pmatrix} \bm \Sigma_{(z\theta)_{t|t-1}} \\ \bm F_t \bm \Sigma_{\theta_{t|t-1}} \end{pmatrix}$
    , and $C_{t|t}= C_{t|t-1} \times g(\mathcal{A}_{\bm y_t})$.
\end{theorem}

The state predictive distribution updating in \eqref{statepred} is analogous to that of the Kalman filter, but the filtering distribution update in \eqref{filtering} has a different form. In the Kalman filter, the filtering distribution parameters are updated by taking the previous best estimate for the states (from the state predictive distribution) and ``correcting'' them using information gained from the new observation. Under the warpDLM, the previous parameters are not corrected; instead, the dimension of the new filtering distribution increases, with the latest data point controlling the bounds over which the latent state can vary. The growing dimension ensures we are retaining all the information about the states, although it also means that sampling can become demanding for longer time series (see Section \ref{sec:comp} for discussion). The selection region $C$ also grows in size with each update, but from a storage perspective this is not burdensome: the region $C_{t|t}$ adds an interval for each component of $\bm y_t$, which simply requires appending a row to two $(t-1) \times n$ matrices. Importantly, the parameter updates depend on the new observation only through $C$, which implies that the system matrices do not need to be updated recursively. This point is better clarified by considering the smoothing distribution.

\subsubsection{Smoothing}
\label{sec:smooth}
The joint smoothing distribution is $p(\bm \theta_{1:T}| \bm y_{1:T})$ where $T$ is the terminal time point. Most commonly, the smoothing distribution in a DLM is obtained by first filtering ``forward" recursively and then smoothing ```backward" in time  \citep{KalmanSmoother}. This strategy targets the marginal smoothing distributions $p(\bm \theta_t | \bm y_{1:T})$ and requires modifications to access the \emph{joint} smoothing  distribution (e.g., \citealp{durbin2002simple}). 

Within the warpDLM, we instead target the joint smoothing distribution directly and analytically. Crucially, this process does not require any preliminary passes through the data: all smoothing parameters can be constructed \emph{a priori} through matrix multiplications. Let $\bm G_{1:t}=\bm G_t \bm G_{t-1}\cdots \bm G_1$ and $
\bm \mu_\theta = (\bm G_{1:1} \bm a_0, \dotsi,  \bm G_{1:T} \bm a_0)$, 
a $pT$-dimensional vector.  Now, let $\bm R_t = \bm G_t \bm R_{t-1} \bm G_t' + \bm W_t$ for $t=1,\ldots, T$.  Then the covariance matrix of $\bm \theta$ is a $pT \times pT$ matrix with $p \times p$ diagonal block entries of 
$
\bm \Sigma_\theta[t,t] = \bm R_t = \bm G_{1:t} \bm R_0 \bm G_{1:t}' + \sum_{q=2}^t \bm G_{q:t} \bm W_{q-1} \bm G_{q:t}' + \bm W_t    
$
and cross covariance entries (for  $t>q$)
$
\bm \Sigma_\theta[t,q] ={\bm \Sigma_\theta[q,t]}^\top = \bm G_{(q+1):t} \bm \Sigma_\theta[q,q].$
We also define two block diagonal matrices, $\boldsymbol{\mathfrak{F}}$ and $\boldsymbol{\mathfrak{V}}$, with diagonal entries of $\bm F_t$ and $\bm V_t$, respectively, for $t=1,\dots, n$, so  $\boldsymbol{\mathfrak{F}}$ is  a $nT \times pT$ matrix and $\boldsymbol{\mathfrak{V}}$ is a $nT \times nT$ matrix. Finally, let $\mathcal{C}_{1:T} =(g(\mathcal{A}_{\bm y_1}), \dotsi,  g(\mathcal{A}_{\bm y_T}))$ be the matrix of selection intervals through time $T$. The joint smoothing distribution is given below.
\begin{theorem}
\label{theorem:smooth}
   Under the warpDLM, the joint smoothing distribution is 
    \begin{equation}
    \label{eq:smoothing}    
        (\bm \theta_{1:T}| \bm y_{1:T}) \sim \mbox{SLCT-N}_{nT,pT} \{\bm \mu_{z} {=} \boldsymbol{\mathfrak{F}} \bm \mu_\theta, \bm \mu_{\theta},  \bm \Sigma_{z} {=} \boldsymbol{\mathfrak{V}}+ \boldsymbol{\mathfrak{F}}\bm \Sigma_\theta  \boldsymbol{\mathfrak{F}}', \bm  \Sigma_{\theta}, 
          \bm  \Sigma_{z \theta} {=} \boldsymbol{\mathfrak{F}} \bm \Sigma_\theta,  \mathcal{C} {=} \mathcal{C}_{1:T} \}
    \end{equation}
\end{theorem}

These results can be applied to the time $s$ filtering  distribution $p(\bm \theta_{1:s}| \bm y_{1:s})$ to provide a more concise representation. Using the analogous smoothing parameters but applied up to time $s$ instead of $T$, 
we can rewrite the parameters of the filtering distribution as $\bm \mu_{z_{s|s}} = \boldsymbol{\mathfrak{F}} \bm \mu_{\theta} \; , \; \bm \mu_{\theta_{s|s}} = \bm G_1^s \bm a_0 \; , \; \bm  \Sigma_{z_{s|s}} = \boldsymbol{\mathfrak{V}}+ \boldsymbol{\mathfrak{F}}\bm \Sigma_\theta  \boldsymbol{\mathfrak{F}}' \;,\; \bm \Sigma_{(z \theta)_{s|s}} = \boldsymbol{\mathfrak{F}} \bm \Sigma_\theta \;,\; \bm \Sigma_{\theta_{s|s}} = \bm R_s.$
Although these covariance terms are quite complex in terms of notation, constructing such matrices is computationally straightforward via matrix multiplications and additions.

The marginal smoothing distribution $p(\bm \theta_{t}| \bm y_{1:T})$ at each time $t$ is readily available from Theorem~\ref{theorem:smooth}. In particular, the SLCT-N family is closed under marginalization, so parameters for the observations stay the same as in the joint case, and we simply pick off the correct block of the state parameters.  This is formalized in the corollary below.
\begin{corollary}
\label{corollary:marginalsmooth}
Under the warpDLM, the marginal smoothing distribution at time $t$ is 
\begin{equation}
    \label{eq:margsmooth}    
    \begin{aligned}
        (\bm \theta_{t}| \bm y_{1:T}) \sim \mbox{SLCT-N}_{nT,p} &\{\bm \mu_{z_{t|T}} = \boldsymbol{\mathfrak{F}} \bm \mu_\theta, \bm \mu_{\theta_{t|T}} = \bm \mu_\theta[t],  \bm \Sigma_{z_{t|T}} = \boldsymbol{\mathfrak{V}}+ \boldsymbol{\mathfrak{F}}\bm \Sigma_\theta  \boldsymbol{\mathfrak{F}}', \bm  \Sigma_{\theta_{t|T}} = \bm \Sigma_{\theta}[t,t],\\
        &  \bm  \Sigma_{(z \theta)_{t|T}} =\bm \Sigma_{z \theta}[,t],  \mathcal{C} = \mathcal{C}_{1:T} \}\;.
    \end{aligned}
    \end{equation}
    where $\bm \Sigma_{z \theta}[,t]$ refers to the $t$-th block of $p$ columns in $\bm \Sigma_{z \theta}$.
\end{corollary}
%The parameters  of  the marginal  smoothing distribution \eqref{eq:margsmooth} have lower dimension than those in the joint smoothing distribution in \eqref{eq:smoothing}. In particular, sampling from \eqref{eq:margsmooth} requires drawing from a $p$-variate normal distribution, as opposed to a $pT$-variate normal for the joint distribution, but both require sampling from an $nT$-variate truncated normal; see Section \ref{subsec:direct}. In practice, the dimension of the truncated normal is the more limiting, and consequently it is often preferable to sample from the joint  distribution.

The joint smoothing distribution also enables direct computation of the warpDLM marginal likelihood: 
\begin{corollary}
\label{corollary:marginallike}
Under the warpDLM, the marginal likelihood is 
\begin{equation}
    \label{eq:marglike}    
    \begin{aligned}
        p(\bm y_{1:T}) = \Bar{\Phi}_{nT}(\mathcal{C}_{1:T}; \bm \mu_{z} = \boldsymbol{\mathfrak{F}} \bm \mu_\theta, \bm \Sigma_{z} = \boldsymbol{\mathfrak{V}}+ \boldsymbol{\mathfrak{F}}\bm \Sigma_\theta  \boldsymbol{\mathfrak{F}}') \;.
    \end{aligned}
    \end{equation}
\end{corollary}
The marginal likelihood  is useful for marginal maximum likelihood estimation of the variance parameters in \eqref{dlm-obs}--\eqref{dlm-evol} and other model comparison metrics.%, for example to select the transformation $g$. 

\subsubsection{Forecasting}
\label{subsubsec:forecast}
A primary advantage  of the warpDLM is  the availability of discrete and correctly-supported forecasting distributions. 
%Just as the marginal likelihood is the normalizing constant of the smoothing distribution, $p(\bm y_{1:(t+1)})$ is the normalizing constant of the filtering distribution at time $t+1$.  
Let $\bm \mu_{z_{t+1|t+1}}$ and $ \bm \Sigma_{z_{t+1|t+1}}$ denote parameters of  $p(\bm y_{1:(t+1)})$ analogous to \eqref{eq:marglike}, and let  $\mathcal{C}_{1:t}$ be the vector of selection intervals up to the known time $t$.  Then the one step forecasting distribution is
\begin{equation}
p(\bm y_{t+1} | \bm y_{1:t}) = \frac{p(\bm y_{1:(t+1)})}{p(\bm y_{1:t})} = \frac{\Bar{\Phi}_{n(t+1)}\{\mathcal{C}_{1:t} \times g(\mathcal{A}_{y_{t+1}}); \bm \mu_{z_{t+1|t+1}}, \bm \Sigma_{z_{t+1|t+1}}\}}{\Bar{\Phi}_{nt}(\mathcal{C}_{1:t}; \bm \mu_{z_{t|t}}, \bm \Sigma_{z_{t|t}})} \;.
\end{equation}
The $h$-step ahead forecasting distribution can be defined similarly. In most applications, the forecasting distribution is defined over all non-negative integers, and thus evaluating it at every point is not possible.  However, in practice most time series of counts take values in a small range (especially locally), and thus the forecasting distribution only has significant mass on a small range of values. For short to medium time series applications, evaluating these probabilities is quite fast and does not pose a computational burden.  

More directly, it is straightforward to simulate from the forecasting distribution given draws from the filtering distribution. Specifically, let $\bm \theta_t^* \sim p(\bm \theta_{t} | \bm y_{1:t})$ denote a draw from the filtering distribution. By passing this draw through the DLM, the evolution equation \eqref{dlm-evol} samples $\bm\theta_{t+1}^{**} \sim N_p(\bm G_{t+1} \bm\theta_{t}^*, \bm W_{t+1}) $   and the observation equation \eqref{dlm-obs} samples $\bm z_{t+1}^* \sim N_n(\bm F_t \bm\theta_{t+1}^{**}, \bm V_{t+1})$, which equivalently yields a draw from $(\bm z_{t+1} | \bm y_{1:t})$. By retaining $ \bm{\tilde y}_{t+1}  = h \circ  g^{-1}(\bm z_{t+1}^*)$, we obtain a draw from the one-step forecasting distribution of the warpDLM; multi-step forecasting proceeds similarly. Hence, the primary computational cost is the draw from the filtering distribution, which is detailed in Section~\ref{sec:comp}.

\subsubsection{Missing Data}
The warpDLM seamlessly handles missing data in any component of $\bm y_t$ at any time $t$. Notably, the  filtering, smoothing, and forecasting distributions depend on $\bm y_t$ only through the implied constraint region (i.e., the collection of intervals). Hence, the absence of an observation $y_{jt}$ implies that the interval, such as $[g(y_{jt}), g(y_{jt} + 1))$, is simply replaced by the interval $(-\infty, \infty)$: there is no information about $y_{jt}$ and therefore no constraint  for the latent $z_{jt}$. As a result, the modifications  for the  filtering, smoothing, and forecasting distributions in the presence of missingness are trivial both analytically and computationally. This result also suggests that the warpDLM is readily applicable to discrete time series data observed at mixed frequencies.

\section{Computing}
\label{sec:comp}
We develop computing strategies for offline and online inference and forecasting. For offline inference, we propose direct Monte Carlo sampling from the relevant filtering or smoothing distributions. We also introduce a Gibbs  sampler for offline inference, which decreases raw computing time for large $T$, but sacrifices some  Monte Carlo efficiency due to the autocorrelated samples. 
Online inference is enabled by an optimal particle filter. The \emph{Monte  Carlo} sampler and the \emph{optimal}  particle filter are uniquely available as a consequence of the results in Section~\ref{sec:theory}. %The supplementary material provides an alternative offline algorithm via Gibbs sampling using a data augmentation strategy similar to probit regression, which can decrease raw computing time when $T$ is large, but  sacrifices Monte Carlo efficiency due to autocorrelation. 

\subsection{Direct Monte Carlo Sampling}
\label{subsec:direct}
The analytic filtering and smoothing distributions for the warpDLM unlock the potential for direct Monte Carlo sampling. This task requires sampling from a selection normal distribution. Crucially, the selection normal distribution admits a constructive representation using  a multivariate normal distribution and a multivariate truncated normal distribution \citep{ArellanoValle2006}, which is converted to a sampler in Algorithm~\ref{algo:slctndraw}. We then apply Algorithm~\ref{algo:slctndraw} to obtain Monte Carlo draws from the warpDLM state predictive distribution \eqref{statepred}, filtering distribution \eqref{filtering}, joint smoothing distribution \eqref{eq:smoothing}, marginal smoothing distribution \eqref{eq:margsmooth}, or forecasting distribution (Section~\ref{subsubsec:forecast}). This computing strategy avoids the need for MCMC, which often requires lengthy simulation runs and various diagnostics. 

\begin{algorithm}
\SetAlgoLined
\KwResult{One sample $\bm \theta$ from SLCT-N distribution}
Given $[\bm \theta ] \sim \mbox{SLCT-N}_{d_1, d_2}(\bm \mu_z, \bm \mu_\theta, \bm \Sigma_z, \bm \Sigma_\theta, \bm \Sigma_{z\theta},  \mathcal{C})$\;
\begin{enumerate}
    \item \textbf{Sample truncated multivariate normal}: Sample $\bm V_0$ from $N_{d_1}(\bm 0, \bm \Sigma_Z)$ truncated to region $\mathcal{C}-\bm \mu_z$
    \item \textbf{Sample multivariate normal}: Sample $\bm V_1$ from $N_{d_2}(\bm 0, \bm \Sigma_\theta - \bm \Sigma_{z\theta}'\bm\Sigma_z^{-1} \bm \Sigma_{z\theta})$
    \item \textbf{Combine results}: Compute $\bm \theta = \bm \mu_\theta +  \bm V_1 + \bm\Sigma_{z\theta}'\bm\Sigma_z^{-1} \bm V_0$ 
\end{enumerate}
 \caption{Sampling from Selection Normal Distribution}
 \label{algo:slctndraw}
\end{algorithm}

These distributions treat the variance components $\bm V_t$ and $\bm W_t$ as fixed and known. Unknown variance components may be accommodated using point estimates (e.g., marginal  maximum likelihood estimators using \eqref{eq:marglike}) or within a blocked Gibbs sampler that iterates draws of the state parameters (Algorithm~\ref{algo:slctndraw}) and the variance components.

The computational bottleneck  of the Monte Carlo sampler is the draw (a) from the multivariate truncated normal.  Efficient sampling  from multivariate truncated normal distributions is an active area of research; we apply the minimax tilting strategy from  \cite{Botev2017}, which is implemented in the R package \verb|TruncatedNormal|. This method is highly efficient and accurate for dimensions up to about $d_1 \approx 500$. For the warpDLM, the dimension $d_1$ corresponds to $nt$ for the filtering distribution at time $t$ or $nT$ for the smoothing distributions. Hence, when the combined data dimension ($nT$) is very large, other computational approaches are needed; we provide these in the subsequent sections.

\subsection{Gibbs sampling}
For lengthy time series, we develop a Gibbs sampler for offline inference that circumvents the computational bottlenecks of direct Monte Carlo sampling. The Gibbs sampler iterates  between a \emph{latent  data augmentation}  draw and a \emph{state parameter} draw (plus the variance components $\bm \psi$, if unknown). This is represented in Algorithm \ref{algo:gibbs}. Although not shown, one can also easily sample from the posterior predictive and forecast distributions for $\bm z_t$ in each pass. Draws from the corresponding distributions for $\bm y_t$ are given by computing $h\left(g^{-1}(\bm z_t)\right)$.
\begin{algorithm}
\begin{enumerate}
    \item \textbf{Sample the latent data}: draw $[\bm z_t| \{ \bm y_t, \bm \theta_t\}_{t=1}^T, \psi]$ from $N_n(\bm F_t \bm \theta_t, \bm V_t)$ truncated to $\bm z_t \in g(\mathcal{A}_{y_t})$ for $t=1,\ldots,T$
    \item \textbf{Sample the states}: draw $[\bm \theta_{1:T}|\bm z_{1:T}, \bm \psi]$ 
    \item \textbf{Sample the parameters}: draw $[\bm \psi| \{ \bm z_t, \bm \theta_t\}_{t=1}^T]$ (from the appropriate full conditionals depending on priors)
\end{enumerate}
\caption{Gibbs Sampler for WarpDLM}
\label{algo:gibbs}
\end{algorithm}

Crucially, the first step of latent data augmentation can be decomposed into $T$ \emph{independent} draws from an $n$-dimensional truncated normal distribution; hence the Gibbs sampling approach doesn't suffer from the same dimensionality problems as the direct sampling method.  Furthermore, the latent state parameter full conditional in step 2 is exactly the smoothing distribution for a Gaussian DLM, and we can thus rely on already-developed Kalman filter-based simulation techniques such as forward filtering backward sampling (FFBS; \citealp{fruhwirth1994data, carter1994gibbs}) or the simulation smoother \citep{durbin2002simple}. Importantly, these sampling steps are \emph{linear} in the length of  the time series $T$. Additionally, both methods have R implementations which can be readily leveraged: the package \verb|dlm| \citep{petrisDLMPackage} performs FFBS and the \verb|KFAS| package \citep{HelskeKFAS} applies the simulation smoother.

In a fully Bayesian framework, we place priors on all variances, and sample from the appropriate full conditionals in step 3.  For univariate models, we put a Uniform$(0,A)$ prior with $A$ large on all standard deviations, leading to relatively simple Gibbs updates. In the multivariate context, we place independent inverse-Wishart priors on variance matrices as in \cite{petrisDLM}.

The sampler for the warpDLM framework is similar to that developed in \cite{Kowal2020a} for the static count regression scenario. The Bayesian dynamic probit sampler \citep{albert1993bayes,Fasano2021} can be regarded as a special case.  In particular, if our data $y_t$ were binary, we could let our transformation $g$ be the identity while defining our rounding operator such that $\mathcal{A}_{y_t=0}=(-\infty,0)$ and $\mathcal{A}_{y_t=1}=(0, \infty)$, thus recovering the probit Gibbs sampling algorithm.

Overall, the Gibbs sampler provides scalable offline inference for the warpDLM in the case  of large $T$, although at the expense of Monte Carlo efficiency; naturally, the direct  Monte Carlo sampler (Algorithm~\ref{algo:slctndraw}) is preferred when the data dimension permits.

\subsection{Optimal Particle Filtering}
\label{subsec:particlefilter}
We design an \emph{optimal} particle filter for online inference with the warpDLM. The key sampling steps  do \emph{not} depend on $T$, which provides scalability for lengthy time series. Online particle filtering algorithms  for state space models usually apply Sequential Importance Sampling with or without resampling  \citep{doucet2000sequential}.  For each time point, sample trajectories or particles are drawn from an importance function, and each particle has an associated importance weight.  For the resampling step, draws from the filtering distribution are obtained via a weighted resampling from the previously drawn particles.  The purpose of this step is to slow the effect of particle degeneration.  

A well-known drawback of particle filtering methods is that the variance of the importance weights grows over time, and thus the trajectories will degenerate, eventually leaving only one particle with nonzero weight.  To avoid quick degeneration, it is important to choose a good importance function; the \emph{optimal} function is that which minimizes the conditional variance of weights. Applying  \cite{doucet2000sequential} to the warpDLM, the optimal importance function is $p(\bm \theta_{t}| \bm \theta_{t-1}, \bm y_{t})$ with weights given by $p(\bm y_{t}| \bm \theta_{t-1})$. For many state space models, this function is not known analytically; however, our results from Section~\ref{sec:theory} produce an exact form for the warpDLM, formalized in the following corollary.

\begin{corollary} 
\label{corollary:pfresults}
Under the warpDLM, we have 
\begin{equation}
\label{eq:importance}
    \begin{aligned}
       (\bm \theta_{t}| \bm \theta_{t-1}, \bm y_{t}) \sim \mbox{SLCT-N}_{n,p} &(\bm \mu_{z} = \bm F_t \bm G_t \bm \theta_{t-1}, \bm \mu_{\theta} = \bm G_t \bm \theta_{t-1},  \bm \Sigma_{z} = \bm V_t+ \bm F_t \bm W_t \bm F_t', \\
        &\bm  \Sigma_{\theta} = \bm W_t,  \bm  \Sigma_{z \theta} = \bm F_t \bm W_t,  C=g(\mathcal{A}_{ \bm y_t}) )
    \end{aligned}
\end{equation}
and 
\begin{equation}
\label{eq:weight}
    \begin{aligned}
       p(\bm y_{t}| \bm \theta_{t-1}) = \Bar{\Phi}_n(g(\mathcal{A}_{y_t}); \bm \mu_{z} = \bm F_t \bm G_t \bm \theta_{t-1}, \bm \Sigma_{z}= \bm V_t+ \bm F_t \bm W_t \bm F_t') 
    \end{aligned}
    \end{equation}
\end{corollary}
We apply Corollary~\ref{corollary:pfresults} to design an \emph{optimal} particle filter for the warpDLM  in  Algorithm \ref{algo:particle}. Crucially, the draw from the selection normal in \eqref{eq:importance} is not dependent on the length of the time series, which provides scalability in $T$.

\begin{algorithm}
\SetAlgoLined
\KwResult{Draws ($\bm \theta_t^{(1)},\ldots,\bm \theta_t^{(S)}$) from filtering distribution $p(\bm \theta_{t}| \bm y_{1:t})$ for $t=1,\ldots,T$}
Given $S$ particle trajectories at time $t=0$ denoted by $\bm \theta_0^{(s)}$ \;
\For{$t\leftarrow 1$ \KwTo $T$}{
    \For{$s\leftarrow 1$ \KwTo $S$}{
        Sample $\Tilde{\bm \theta}_t^{(s)}$ from \eqref{eq:importance} conditional on $\bm \theta_{t-1} = \bm \theta_{t-1}^{(s)}$ \;
        Calculate weight up to normalizing constant using \eqref{eq:weight}: $w_t^{(s)} \propto p(\bm y_{t}| \bm \theta_{t-1}^{(s)})$
    }
    Normalize importance weights: $\Tilde{w_t}^{(s)} = \frac{w_t^{(s)}}{\sum_{s=1}^S w_t^{(s)}}$ \;
    Obtain final draws $\bm \theta_t^{(1)},\ldots,\bm \theta_t^{(S)}$ by resampling $ \Tilde{ \bm \theta}_t^{(1)},\ldots, \Tilde{ \bm \theta_t}^{(S)}$ with weights $\Tilde{w_t}^{(s)}$
 }
 \caption{Optimal Particle Filter for warpDLM}
 \label{algo:particle}
\end{algorithm}

%Particle filtering algorithms output weighted samples, and thus any summarization must take into account these weights. Here we resample at every time step, which gives equally weighted draws but also can decrease the diversity of the sampled state trajectories. Future research could look at the effect on algorithm performance of different schemes, e.g., stratified sampling or residual resampling.  

The (unnormalized) weights can be used to estimate the marginal likelihood $p(\bm y_{1:T})$:
\begin{equation}
\hat{p}(\bm y_{1:T}) \coloneqq \hat{p}(\bm y_{1}) \prod_{t=2}^T \hat{p}(\bm y_{t}|\bm y_{1:t-1}) 
\end{equation}
where $\hat{p}(\bm y_{1})=\frac{1}{S}\sum_{s=1}^S w_1^{(s)}$ and $\hat{p}(\bm y_{t}|\bm y_{1:t-1}) = \frac{1}{S}\sum_{s=1}^S w_t^{(s)}$. As with equation \eqref{eq:marglike}, this can be used for model comparisons or estimation of unknown parameters \citep{ParticleParameter2015}. Such an approximation is also key when designing more complex sequential Monte Carlo schemes or sampling based on particle methods \citep{Andrieu_Doucet_Holenstein_2010}.

As both the cross-section dimension $n$ and the length of the time series $T$ increase, we expect that the quality of the particles will decline, despite the optimality of our transition density. This effect is typically more pronounced for the smoothing distribution than the filtering distribution. 
In that case, we may seek to improve the particle filter (Algorithm~\ref{algo:particle}) by applying relevant techniques from sequential Monte Carlo sampling, such as a mutation step, or resort to particle smoothing \citep{doucet2009tutorial} when $T$ is especially large. However, we show empirically that the proposed particle filter  is capable of accurate inference for $n=5$ and $T=4000$ (see the supplementary material).

%Particle filtering is compatible with offline analyses, which may be used to provide  parameter estimates and initializations for the (online) particle filter. We illustrate this kind of tandem analysis in Section~\ref{subsec:realdata}.

\section{Warped DLMs in Action}
\label{sec:application}
In this section, we showcase several facets of the warpDLM framework.  We first perform a forecasting analysis on simulated data, showing that the warpDLM offers better distributional forecasting than competing Bayesian methods.  We then apply the warpDLM to model a bivariate time series of drug overdose counts  in Cincinnati, and highlight both the offline and online computing capabilities.

%we briefly discuss evaluation of distributional forecasts for discrete data.

\subsection{Evaluating Probabilistic Forecasts}
A prominent advantage of the warpDLM is that it provides a forecasting \emph{distribution} that matches the (discrete) support of the data. Distributional forecasts are especially important for discrete data, whereas point forecasts are less informative. Evaluation of probabilistic forecasts should first and foremost verify  that forecast distributions are \textit{calibrated} in the sense that observed values are consistent with samples from the distribution \citep{czado2009predictive}. To assess calibration, we apply the randomized Probability Integral Transform (rPIT). Let $H_t$ denote the true (discrete)  cumulative forecasting distribution, from which a count $y_t$ is observed.  Each model produces an estimate $\hat{H}_t$ of the forecasting distribution. The rPIT is drawn from
$\Tilde{p}_t \stackrel{iid}{\sim} \mbox{Uniform}[\hat{H}_t(y_t -1), \hat{H}_t(y_t)].
$
If the model is well-calibrated, then $\hat{H}_t \approx H_t$ and consequently $\Tilde{p}_t \stackrel{iid}{\sim} \mbox{Uniform}(0,1)$. Thus, forecasts  are evaluated by (i) drawing $\Tilde{p}_t$ for each $t$ and (ii) comparing the draws to a standard uniform distribution. 
%Thus, we can evaluate our method by forecasting at various time points and drawing an rPIT value for each forecast, then plotting these values to see how they compare to the standard uniform distribution.

Among calibrated forecasts, we evaluate the \textit{sharpness} of the forecasting distribution, such as the tightness of the prediction intervals.  \cite{czado2009predictive} propose to evaluate the sharpness via scoring rules, in particular the logarithmic scoring rule 
$-\log(p_x)$ where $p_x$ is the probability mass of the predictive distribution at the observed count $x$.  We will utilize these diagnostics to evaluate various warpDLMs and DGLM competitors.

\subsection{Forecasting on Simulated Data}
\label{subsec:sim}
Existing models for count time series are often unable to capture multiple challenging discrete distributional features, such as zero-inflation, boundedness, and heaping. By comparison, the warpDLM is particularly well-suited for these features. At the same time, the warpDLM is sufficiently flexible to model simpler count data settings for which existing models have been designed, such as Poisson data.  To showcase this breadth of modeling, we designed simulation experiments for both scenarios and compared the forecasting performance of various warpDLMs against Bayesian competitors.

In the first experiment, we simulated zero-inflated and bounded count time series data. One real-world example of such data is analyzed in \cite{Ensor_Ray_Charlton_2014}, where the quantity of interest is the number of hours per day which exceed a certain pollutant threshold. Many days show no pollution, leading to zero-inflation, but occasionally every hour in a day exceeds the threshold, resulting in heaped values at the upper bound of 24. Specifically, we sampled from a zero-inflated Poisson with time varying mean $\lambda_t$, where $\lambda_1 \sim \mbox{Uniform}(5, 15)$ and $\lambda_{t+1} \sim N(\lambda_t, 0.2)$ for $t=2,\ldots,T$. The zero-inflation component was drawn uniformly from $[0.1, 0.3]$. Any sampled value above 24 was rounded down to the bound. 

In the second scenario, we simulated data from the INGARCH class using the R package \verb|tscount| \citep{tscountVig}.  An INGARCH($p,q$) model assumes $y_t \sim \mbox{Poisson}(\lambda_t)$ with
$\lambda_t = \beta_0 + \sum_{k=1}^p \beta_k Y_{t-k} +  \sum_{\ell=1}^q \alpha_\ell \lambda_{t-\ell}.$ 
%Note that various stationarity constraints apply to the coefficients and limit the values that can be taken (see \cite{Fokianos_Rahbek_Tjostheim_2009} for details).  
In our simulation, we chose relatively simple dynamics of $p=1$ and $q=1$ and set $\beta_0=0.3$, $\beta_1 = 0.6$, and $\alpha_1=0.2$, with coefficients selected such that simulated series were low-count time series.  

For both simulation scenarios, 30 time series of length $T=200$ were simulated. Among all competing state space models, we applied the same evolution equation: a local level DLM with the univariate state $\theta_t$ and the time-invariant system $\{1, 1, V, W\}$ with initial state $\theta_0 \sim N(0,3)$. The warpDLMs included \emph{parametric} transformations (identity and square-root) and the \emph{nonparametric} transformation \eqref{g-hat}. In each case, the variance components $V$ and $W$ were estimated via maximization of the marginal likelihood \eqref{eq:marglike} and then Monte Carlo samples from the joint smoothing distribution \eqref{eq:smoothing} were drawn using Algorithm \ref{algo:slctndraw}. These draws were used to simulate from the forecasting distribution, as outlined in Section~\ref{subsubsec:forecast}. Competing methods included Poisson, negative binomial, and (for the bounded data scenario) binomial DGLMs, as implemented in the R package \verb|KFAS| \citep{HelskeKFAS}. We estimated variance components with maximum likelihood and used the \verb|KFAS| default value of the dispersion parameter for the negative binomial DGLM. 

For the DGLM models, it took several attempts to identify an initial value that resulted in convergence for the maximum likelihood estimation.  In general, the likelihoods of state space models often have complex landscapes which can make parameter estimation difficult and sensitive to starting inputs \citep{HelskeKFAS}. For the warpDLM methods, a suitable starting point was found by running a a Gibbs sampler for 2000 iterations (treating the first 500 as burn-in) and then using the posterior means as the input to the optimization algorithm.

%The zero-inflated Poisson simulation setting is an example of upper-bounded data, and thus could be modeled with a binomial DGLM.  We estimated this model and performance was shown to be uniformly worse than every other competitor, with forecast scores on average twice as worse as the nearest model.  For this reason, it was excluded in the results reported here, but the corresponding graphs with the binomial DGLM included can be found in the supplement.

In order to evaluate the forecasting performances, we used time series cross-validation for one-step-ahead predictions.  Starting at $t=100$, each of the five methodologies (three warpDLM methods and two DGLM methods) were fitted and one-step-ahead forecasts were drawn. For 50 equally dispersed time points from $t=100$ to $t=200$, the process of training and drawing one-step ahead forecasts was repeated.  Thus, we have 50 observations to compare to 50 one-step-ahead forecasts for each of the 30 time series. 

Since we evaluate a multitude of series, we follow \cite{Kolassa_2016} in testing the uniformity of the rPIT values using the  data-driven smooth test of \cite{Ledwina_1994}. Hence, the output is a p-value, where small values indicate poor  calibration. 
We then consider the sharpness by evaluating the logarithmic scoring rule. Our output for these models is 5000 draws from the predictive distribution, so we approximate the probability mass by calculating the proportion of draws which predicted that value. If a value was not sampled in any posterior draw, but was observed, then the logarithmic scoring rule would be infinite.  To avoid this, we instead set the mass in that case to 0.0001. For each series, we find the mean log score across the 50 time points and compute the percent difference over a baseline method, chosen to be the Poisson DGLM. The logarithmic scoring rule is negatively oriented, so a negative percent difference equates to improved forecasting.

Figure \ref{fig:ZIP} presents the calibration and sharpness across simulations for the zero-inflated and bounded Poisson example. The binomial DGLM exhibited extremely poor forecasting performance and is omitted for clarity of presentation; results including the binomial DGLM are in the supplement.  A clear takeaway is that the warpDLM with nonparametric  transformation is well-calibrated. The identity warpDLM is also reasonably well-calibrated for many simulations, highlighting the importance of the rounding operator. The competing models produce extremely low p-values for almost every simulation, which decisively shows that these models are not well-calibrated;  \cite{czado2009predictive} argue that such poor calibration should automatically disqualify these models. The results for sharpness are consistent: the warpDLM with nonparametric transformation produces the best distributional forecasts (on average 30\% better than the Poisson DGLM), while the warpDLM with identity transformation is the second best. Modifications for the Poisson and negative binomial DGLMs to  enforce the upper bound of 24 on forecasts did not substantively change the results. Clearly, the warpDLM offers significant forecasting gains over competing methods, and advertises both good calibration and maximal sharpness.

\begin{figure}
    \centering
    \includegraphics[width=\linewidth,height=\textheight,keepaspectratio]{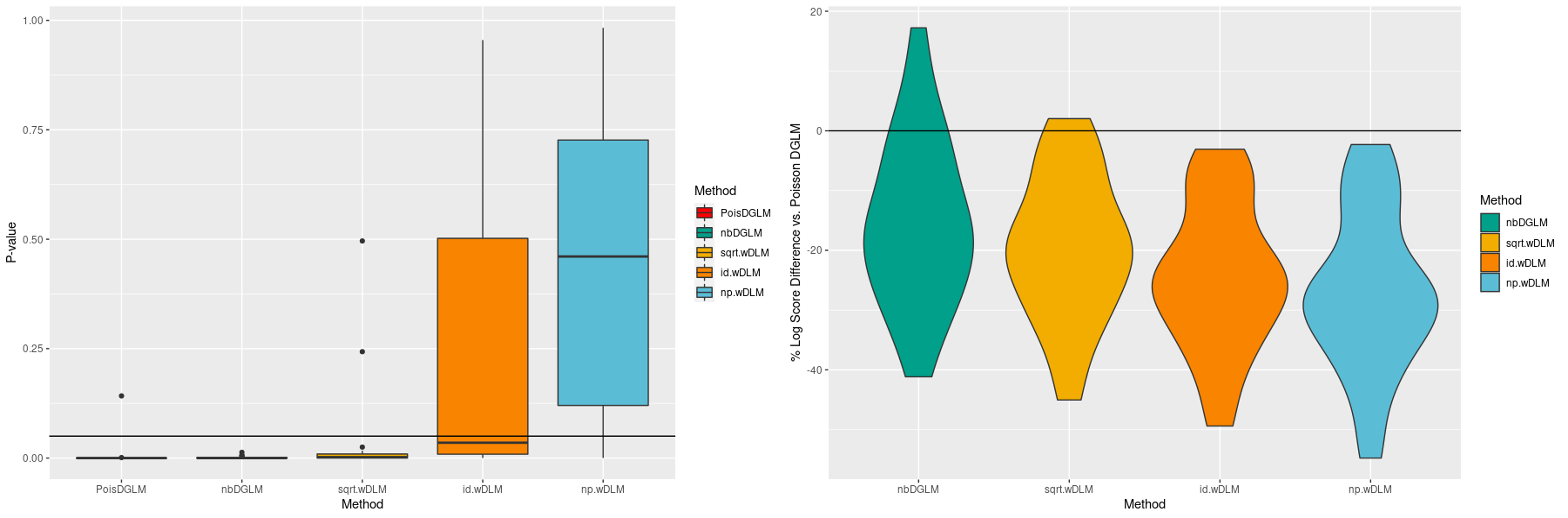}
    \caption{Zero-inflated Poisson simulations; Left: p-values measuring calibration (larger is better) across simulations, with a line at p=0.05; Right: percent difference in log score compared to baseline (negative values imply improved forecasting relative to Poisson DGLM)}
    \label{fig:ZIP}
\end{figure}

The second simulation scenario is presented in Figure~\ref{fig:INGARCH}. This simpler scenario should be more favorable for the Poisson DGLM. Nonetheless, the warpDLM has the flexibility to match or improve upon the DGLM forecasts. All five models show proper calibration for nearly all simulations. Each of warpDLMs exhibits sharpness similar to the Poisson DGLM, and the warpDLM with square-root transformation actually appears to improve on the forecasting for the majority of the series. The negative binomial DGLM performs uniformly worse than all other methods.

\begin{figure}
    \centering
    \includegraphics[width=\linewidth,height=\textheight,keepaspectratio]{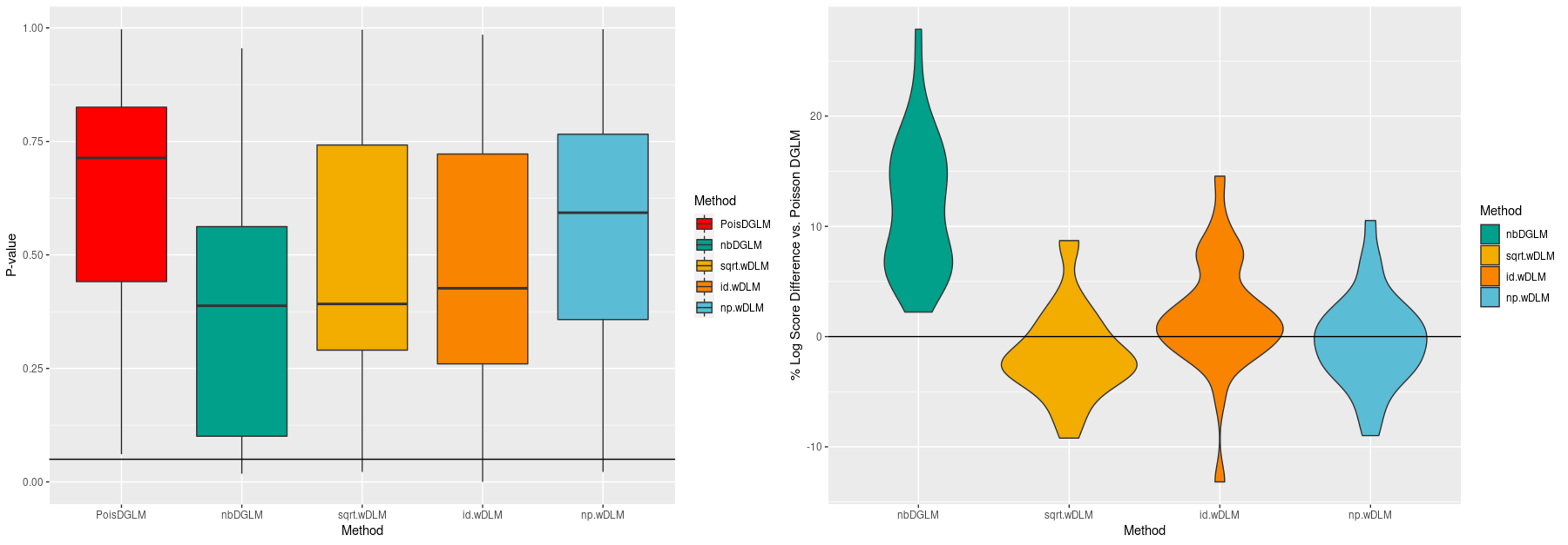}
    \caption{INGARCH simulations; Left: p-values measuring calibration (larger is better) across simulations, with a line at p=0.05; Right: percent difference in log score compared to baseline (negative values imply improved forecasting relative to Poisson DGLM)}
    \label{fig:INGARCH}
\end{figure}

Additional simulation results for higher-dimensional data ($n=5$, $T = 4000$) are included in the supplementary material. 

\subsection{Real Data Analysis}
\label{subsec:realdata}
%The results of the previous section show that the warpDLM framework is capable of modeling multiple data-generating processes and improves over competing methods in terms of forecasting performance. In this section, we use a real-world dataset to illustrate how the warpDLM can handle multivariate time series, a challenge for most count time series models. Furthermore, we display the adaptability of computation for warpDLMs, demonstrating how analysts can perform both offline and online inference. 
We apply the warpDLM to a multivariate time series of drug overdose (OD) counts in Cincinnati. The OD reports are contained within a dataset of all fire and EMS incidents in the   \href{https://data.cincinnati-oh.gov/Safety/Cincinnati-Fire-Incidents-CAD-including-EMS-ALS-BL/vnsz-a3wp}{city's open access data repository}, and include daily counts of EMS calls for heroin and non-heroin ODs from January 1st, 2018 to January 31st, 2021; see the supplement for data wrangling and cleaning details. Notably, Cincinnati has struggled with a severe problem of drug ODs in the city, and in particular heroin \citep{li2019suspected}. From a statistical modeling perspective, these data are discrete, time-ordered, multivariate ($n=2$), and moderately lengthy ($T = 1127$); see Figure~\ref{fig:bivariateTS}.

\begin{figure}
    \centering
    \includegraphics[width=\linewidth]{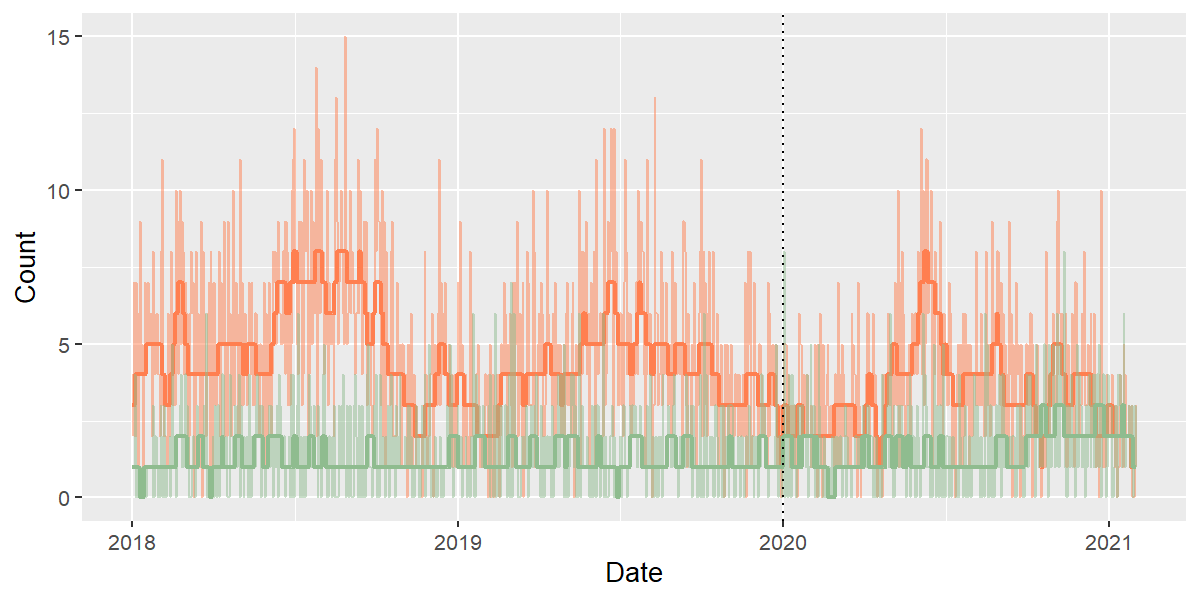}
    \caption{Bivariate time series of daily counts of EMS calls for heroin (larger values) and non-heroin ODs. The dotted vertical line represents the transition between the online and offline datasets. The solid lines are the median smoothed predictions from the warpDLM.}
    \label{fig:bivariateTS}
\end{figure}

We specify a warpDLM with the nonparametric transformation from \eqref{g-hat}, the default rounding operator, and a linear growth model for each series $i\in\{1,2\}$:
\begin{align}
\label{eq:SUTSE_observation}
z_{i,t} &= \mu_{i,t} + v_{i,t} \\
\label{eq:SUTSE_level}
\mu_{i,t} &= \mu_{i,t-1} + \beta_{i,t-1} + w_{i,t-1}^\mu \\
\label{eq:SUTSE_slope}
\beta_{i,t} &= \beta_{i,t-1} + w_{i,t-1}^\beta
\end{align}
where $\mu_{i,t}$ is a local level, $\beta_{i,t}$ is a slope, and \eqref{eq:SUTSE_level}--\eqref{eq:SUTSE_slope} comprise the DLM evolution equation \eqref{dlm-evol}. The linear growth model allows for more flexibility than the local level specification previously introduced, but also can effectively reduce to the simpler local level if the data supports it (i.e. if the slope variance parameters are estimated to be very small).  Dependence between the $n=2$ series is induced by the observation error distribution $\bm v_{t} \sim N_2\left(\bm 0, \bm V\right)$ and the state error distributions  $\bm w_{t}^\mu \sim N_2\left(\bm 0, \bm W_\mu \right)$ and  $\bm w_{t}^\beta \sim N_2\left(\bm 0, \bm W_\beta \right)$ for $\bm v_{t} = (v_{1,t}, v_{2,t})'$, 
$\bm w_{t}^\mu = (w_{1,t}^\mu, w_{2,t}^\mu)'$, and
$\bm w_{t}^\beta = (w_{1,t}^\beta, w_{2,t}^\beta)'$, with $\bm V, \bm W_\mu, \bm W_\beta$ assigned inverse-Wishart priors.  Model \eqref{eq:SUTSE_observation}--\eqref{eq:SUTSE_slope} is expressed in the usual DLM form of \eqref{dlm-obs}--\eqref{dlm-evol} in the supplement. As an alternative approach, we also constructed a seasonal warpDLM, but its performance was indistinguishable from the linear growth model (see the supplementary material).

We highlight the capabilities of the warpDLM for both offline and online analysis of multivariate count time series data. Specifically, we perform an offline analysis for the years of 2018 and 2019 ($T=730$) and switch to an online particle filter for the 2020 and January 2021 data ($T=397$). Given the lengthy time series, offline inference was conducted using the Gibbs sampler (Algorithm \ref{algo:gibbs}), which returned 10,000 draws from the state filtering distribution of $\bm \theta_t$ for $t=730$ and the posterior means of the variance components. These quantities were input into the particle filter of Algorithm \ref{algo:particle} and run sequentially through each of the $T=397$ time points in the online dataset, resulting in 10,000 sequential paths of the smoothed states.  These smoothed states were iterated one step ahead to obtain draws from the one-step forecasting distribution as in Section \ref{subsubsec:forecast}. We also considered longer forecast horizons (see the supplementary material). Unsurprisingly, the forecast distribution accuracy decreased for longer horizons.

To summarize the trends in Figure~\ref{fig:bivariateTS}, we compute the pointwise medians of the smoothing predictive distributions $p(\Tilde{\bm y}_{t} | \bm y_{1:T})$ under the warpDLM. For each sampled path of states $\bm \theta_{1:T}^* \sim p(\bm \theta_{1:T} | \bm y_{1:T})$, we draw a corresponding path $\Tilde{\bm z}_{1:T}$ via the DLM observation equation \eqref{dlm-obs}, and set $ \bm{\tilde y}_{1:T}  = h \circ g^{-1}(\Tilde{\bm z}_{1:T})$, similar to the forecasting procedure in   Section~\ref{subsubsec:forecast}. Notably, this quantity provides a count-valued  point estimate that ``smooths" each time series conditional on the complete data $\bm y_{1:T}$. The heroin ODs exhibit both higher values and greater variability over time, although the  two time series appear to converge to similar levels in early 2021. 

To evaluate the fitness of the warpDLM, we report calibration of the one-step forecasts in Figure~\ref{fig:rPIT}. The sorted rPIT values of both series are plotted against standard uniform quantiles; the 45 degree line indicates calibration.  Since the rPIT values are random quantities, some variation is natural. We illustrate this inherent variability by sampling 100 draws  of size $396$ from a standard uniform distribution, and plot their sorted values against the uniform quantiles.
\begin{figure}
    \centering
    \includegraphics[width=.9\linewidth,height=\textheight,keepaspectratio]{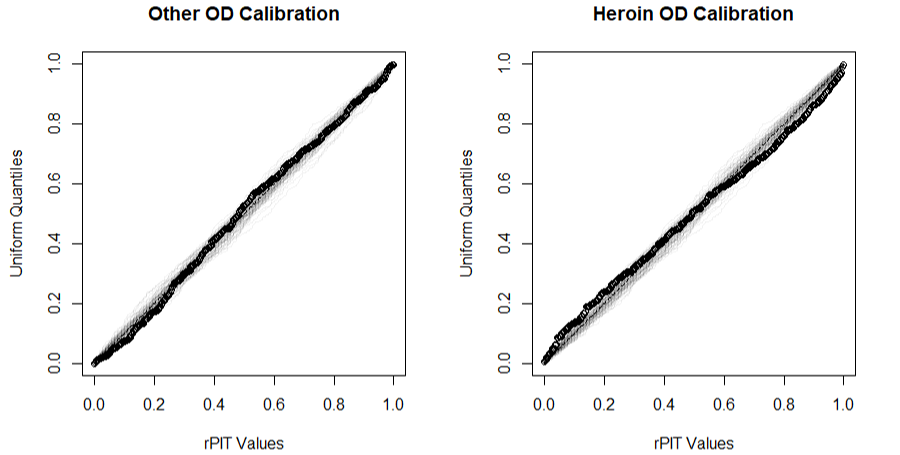}
    \caption{rPIT plots from one-step-ahead particle filter forecasts for both OD count series. The dashed 45 degree line indicates perfect calibration. Gray lines are computed from standard uniform draws, to illustrate inherent variability in the randomized metric.}
    \label{fig:rPIT}
\end{figure}

The warpDLM one-step forecasts are well-calibrated for non-heroin ODs. For the heroin ODs, the calibration shows minor deviations from uniformity. The histogram of the rPIT values (not shown here) shows a slightly inverse-U shape, which implies forecasts are overdispersed relative to the actual values. A closer investigation of the data shows that the online dataset showed considerably less variation than the offline dataset. Since the nonparametric transformation \eqref{g-hat} was inferred using only the \emph{offline} dataset, this result suggests the need to occasionally update the transformation during online inference, for example by updating the estimate of the marginal CDF $F_y$.

Lastly, we evaluate the computational performance of the optimal particle filter, both for Monte Carlo efficiency and raw computing  time. Figure~\ref{fig:ESS} presents the effective sample size, $ESS = 1/ \sum_{s=1}^S (\Tilde{w_t}^{(s)})^2$ using the normalized weights $\Tilde{w_t}$, along with the number of seconds needed to update the model at each time point (on a laptop with an Intel Core i5-6200U CPU with 8 GB RAM). The \emph{optimal} particle filter---uniquely obtained via the warpDLM distributional results from Section~\ref{sec:theory}---produces consistently large ESS values. Crucially, the computation time does \emph{not} increase with the time index and hovers around 13-14 seconds for most updates. Given this computation time, the warpDLM may be applied to more demanding streaming data, such as minute-by-minute data. In all, the optimal particle filter for the warpDLM enabled (i) modeling of multivariate  count time series  data, (ii)  well-calibrated forecasts, and (iii) sufficient scalability for moderate- to high-frequency online analysis.

\begin{figure}
    \centering
    \includegraphics[width=.9\linewidth,height=\textheight,keepaspectratio]{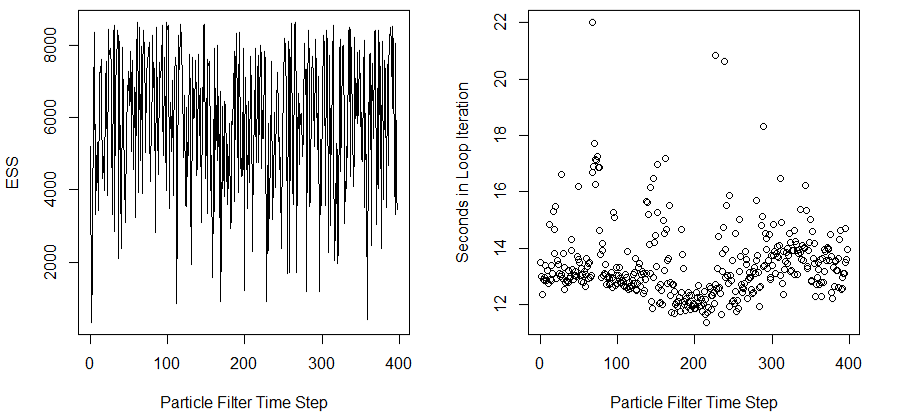}
    \caption{Effective sample sizes (left; maximum is 10,000) and raw computing time in seconds (right) for each time point in the online dataset.}
    \label{fig:ESS}
\end{figure}

\section{Conclusion}
\label{sec:conc}
We introduced a state space modeling framework for a multivariate  time series of counts by \emph{warping} a dynamic  linear model  (warpDLM). The  warpDLM advertises (i) a data-coherent model (and forecasting distribution) that matches the discrete support of the data; (ii) semiparametric modeling capabilities via an inferred transformation; (iii) the familiarity and dynamic flexibility of DLMs, including the ability to handle missing data; (iv) analytic filtering and smoothing recursions; and (v) customized algorithms for forecasting and inference, including direct Monte Carlo sampling and a Gibbs sampler for offline inference and an optimal particle filter for online  inference. These results and methods also apply for integer-valued and rounded data.

Using simulated data, we showed that the warpDLM offers better distributional forecasting than Bayesian competitors, particularly when the count data exhibit multiple complexities such as zero-inflation and boundedness.  Finally, we showcased the online inference and forecasting capabilities of the warpDLM  for a multivariate time series of  drug overdose  counts.  

The generality of the warpDLM enables several useful extensions. Successful tools from (Gaussian) DLMs, such as scale-free variance modeling or discount factors,  may  be incorporated into the warpDLM framework. These specifications  often result in $t$-distributed state updates and predictions, which may be linked to our results via the selection-$t$ distribution. Similarly, other increases in modeling complexity can be accompanied by appropriate algorithmic advancements, such as sequential Monte Carlo methods that sample both the unknown parameters and the states concurrently \citep{SMC2_2012}. Despite the inherent challenges in these more complex modeling and computing environments, the analytic and recursive updates derived for the warpDLM offer a promising pathway for efficient and perhaps optimal implementations.  

The warpDLM can also be thought of as a (nonlinear) hierarchical state-space model (HSSM; \citealp{gamerman1993dynamic}), with the latent data $z_t$ introducing a first-level hierarchy. In particular, connections may be made to the very general HSSM structure presented in \cite{katzfuss2020ensemble}, with the warping operation presented here acting as the transformation layer in their setup. Exploring this relationship further might open up new possibilities for applying the warpDLM in high-dimensional problems.

\bibliographystyle{apalike}
\bibliography{refs.bib}

%%%%%%%%%%%%%%%%%%%%%%%%%%%%%%%%%%%%%%%%%%%%%%%%%%%%%%%%%%%%%%%%%%%%%%%%%%%%%%%%%%%%%%%%%%%%%%%%%%%%%%%%%%%%%%%%%%%%%%%%%%%%%%%%%%%%%%%%%%%%
%Begin Supplement
%Restart numbering 
\clearpage
\pagenumbering{arabic}
\begin{appendices}
\counterwithin{figure}{section}

\noindent

This supplementary document contains the following:
\begin{itemize}
    \item Section~\ref{sec-proofs}: Proposition and proof of SLCT-N closure under linear transformations as well as proofs of main results in the paper
    \item Section \ref{sec:simextra}: Results from simulations that include binomial DGLM
    \item Section \ref{sec:appdeets}: Further details on application dataset and additional model results
    \item Section \ref{sec:particlesim}: Results from particle filter applied to higher-dimensional simulated dataset
\end{itemize}

\section{Further Theory and Proofs}
\label{sec-proofs} 
As discussed in Section~\ref{subsec:selectionbasics}, the selection normal inherits closure under linear transformations, a result which is used to show that the warpDLM state predictive distribution stays within the selection normal family.  We formalize and prove this in the below proposition.
\begin{proposition}
\label{prop:linearity}
    Suppose $\bm \theta \sim \mbox{SLCT-N}_{n, p}(\bm \mu_z, \bm \mu_\theta, \bm \Sigma_z, \bm \Sigma_\theta, \bm \Sigma_{z\theta},  \mathcal{C})$ and let $\bm \theta^* \stackrel{d}{=} \bm A \bm\theta + \bm a$ where $\bm A$ is a fixed $q \times p$ matrix and $\bm a \sim N_q(\bm \mu_a, \bm \Sigma_a)$ is independent of $\bm \theta$. 
    
    Then $\bm \theta^* \sim \mbox{SLCT-N}_{n, q}(\bm \mu_z, \bm \mu_{\theta^*}, \bm \Sigma_z, \bm \Sigma_{\theta^*}, \bm \Sigma_{z\theta^*},  \mathcal{C})$ where $\bm \mu_{\theta^*} =\bm A\bm \mu_\theta + \bm \mu_a$, $\bm \Sigma_{\theta^*}= \bm A \bm \Sigma_\theta \bm A' +\bm \Sigma_a$, and $\bm \Sigma_{z\theta^*} = \bm \Sigma_{z\theta} \bm A'$.
\end{proposition}
\begin{proof}
    First consider the distribution of $\bm A \bm \theta$. Since the moments of the joint distribution $(\bm z, \bm \theta)$ are given by assumption, it follows that $(\bm z, \bm A \bm \theta)$ is jointly Gaussian with moments available by straightforward calculation, so $\bm A \bm \theta \sim \mbox{SLCT-N}_{n, q}(\bm \mu_z, \bm A\bm \mu_\theta, \bm \Sigma_z, \bm A \bm \Sigma_{\theta}\bm A', \bm \Sigma_{z\theta} \bm A',  \mathcal{C})$. 
    
    As noted in \citesupp{ArellanoValle2006}, selection normal distributions have a moment generating function (MGF).  For $[\bm \theta | \bm z \in \mathcal{C}] \sim \mbox{SLCT-N}_{n, p}(\bm \mu_z, \bm \mu_\theta, \bm \Sigma_z, \bm \Sigma_\theta, \bm \Sigma_{z\theta},  \mathcal{C})$, the MGF is:
    \begin{equation*}
    \label{mgf-slct-n}
    M_{[\bm \theta | \bm z \in \mathcal{C}]}(\bm s) = \exp\left(\bm s' \bm \mu_\theta + \frac{1}{2}\bm s' \bm\Sigma_\theta\bm s\right) \frac{\bar\Phi_n(\mathcal{C}; \bm\Sigma_{z\theta}\bm s + \bm \mu_z, \bm \Sigma_z)}{\bar\Phi_n(\mathcal{C}; \bm \mu_z, \bm \Sigma_z)}.
    \end{equation*}
    
    Using the selection normal MGF and noting independence between $\bm A \bm\theta$ and $\bm a$, it follows that $M_{\bm \theta^*}(\bm s) = M_{ \bm A \bm \theta+\bm a}(\bm s)= M_{ \bm A \bm \theta}(\bm s) M_{\bm a}(\bm s)$ where $ M_{ \bm A \bm \theta}$ is given by inserting the appropriate parameters into the equation above and $M_{\bm a}(\bm s) = \exp(\bm s' \bm \mu_a + \frac{1}{2}\bm s' \bm \Sigma_a \bm s)$. The result of this product is the MGF of the stated SLCT-N distribution.
\end{proof}

\noindent
We now state the proofs of all results in the main article, with the exception of Theorem \ref{normal-conjugacy}, whose result follows directly from the definition of a selection normal distribution and the warpDLM model setup. The concurrent work of \citesupp{kowalSTARconjugate} also explores Theorem \ref{normal-conjugacy} and Lemma \ref{lem:conjugacy}, but with focus on non-dynamic linear regression for discrete data. 

\begin{proof}[Proof (Lemma \ref{lem:conjugacy})]
    The SLCT-N prior is equivalently defined by $[\bm \theta | \bm z_0 \in \mathcal{C}_0]$ for $(\bm z_0', \bm \theta')'$ jointly Gaussian with moments given in the prior. The posterior is constructed similarly: $[\bm \theta | \bm y] \stackrel{d}{=} [\bm \theta | \bm z_0 \in \mathcal{C}_0, \bm z \in g(\mathcal{A}_{\bm y})] \stackrel{d}{=} [\bm \theta | \bm z \in \mathcal{C}_0 \times g(\mathcal{A}_{\bm y})]$ where $\bm z_1 = (\bm z_0', \bm z')'$. It remains to identify the moments of $(\bm z_1', \bm \theta')' = (\bm z_0', \bm z', \bm \theta')'$. For each individual block of $\bm z_0$, $\bm z$, and  $\bm \theta$ and the pairs ($\bm z_0, \bm \theta$) and ($\bm z, \bm \theta$), the moments are provided by either the prior or the posterior in Theorem \ref{normal-conjugacy}. Finally, we have cross-covariance $\mbox{Cov}(\bm z_0, \bm z) = \mbox{Cov}(\bm z_0, \bm F \bm \theta + \bm v) =  \bm \Sigma_{z_0\theta}\bm F'$. % since  $\bm z_0$ is independent of $\bm \epsilon$.
\end{proof}

\begin{proof}[Proof (Theorem \ref{theorem:filtering})]
From the model setup, we have $\bm \theta_t =\bm  G_t \bm \theta_{t-1} + \bm w_t$ and $\bm \theta_{t-1} | \bm y_{1:t-1}$ follows a SLCT-N distribution as given.  Thus by Proposition \ref{prop:linearity}, we have our result that $\bm \theta_{t} | \bm y_{1:t-1}$ follows a SLCT-N with the given parameters.  

Then, with the state predictive distribution $\bm \theta_{t} | \bm y_{1:t-1}$ as our prior combined with the warpDLM likelihood \eqref{likelihood-time}, we can directly apply Lemma \ref{lem:conjugacy} to get the second result. 
\end{proof}

\begin{proof}[Proof (Theorem \ref{theorem:smooth})]
We have that $p(\bm \theta_{1:T}| \bm y_{1:T}) \propto p(\bm \theta_{1:T})p(\bm y_{1:T} | \bm \theta_{1:T})$.  By the model specification, we know $\bm \theta_{1:T} \sim N_{pT}(\bm \mu_\theta, \bm \Sigma_\theta)$ as defined above.  Furthermore, recall that given the latent states, the observations are conditionally independent.  Hence we can represent the likelihood as $p(\bm y_{1:T} | \bm \theta_{1:T}) = \prod_{t=1}^T \Bar{\Phi}_n(g(\mathcal{A}_{y_t}); \bm F_t\bm \theta_t, \bm V_t)$ which can be rewritten as $\Bar{\Phi}_{nT}(C; \boldsymbol{\mathfrak{F}}\bm \theta_{1:T}, \boldsymbol{\mathfrak{V}})$.  Therefore, we have
$$
p(\bm \theta_{1:T}| \bm y_{1:T}) \propto \phi_{pT}(\bm \theta_{1:T} ; \bm \mu_\theta, \bm \Sigma_\theta) \Bar{\Phi}_{nT}(C; \boldsymbol{\mathfrak{F}}\bm \theta_{1:T}, \boldsymbol{\mathfrak{V}})
$$
which is the kernel of a SLCT-N distribution with parameters given in \eqref{eq:smoothing}.
\end{proof}

\begin{proof}[Proof (Corollary \ref{corollary:marginallike})]
The marginal likelihood is simply the normalizing constant of the joint smoothing distribution, and its form is given by the denominator in \eqref{density-slct-n}. 
\end{proof}

\begin{proof}[Proof (Corollary \ref{corollary:pfresults})]
By Bayes' Theorem, we have
\begin{equation}
\label{pf-bayes}
p(\bm \theta_{t}| \bm \theta_{t-1}, \bm y_{t}) = \frac{p(\bm y_{t}| \bm \theta_{t})p(\bm \theta_{t}| \bm \theta_{t-1})}{p(\bm y_{t}| \bm \theta_{t-1})} \;.    
\end{equation}
Using the DLM and warpDLM properties, we know that
$$
p(\bm y_{t}| \bm \theta_{t}) = \Bar{\Phi}_n(g(\mathcal{A}_{y_t}); \bm F_t \bm \theta_t, \bm V_t) \quad \mbox{and} \quad p(\bm \theta_{t}| \bm \theta_{t-1}) = \phi_p(\bm G_t \bm \theta_{t-1}, \bm W_t) \;.
$$

Putting this together, recognize that the numerator of \eqref{pf-bayes} forms the kernel of a selection normal distribution \eqref{density-slct-n} with parameters given in \eqref{eq:importance}.  Furthermore, the marginal distribution $p(\bm y_{t}| \bm \theta_{t-1})$ is that kernel's normalizing constant, leading us to the result of \eqref{eq:weight}.
\end{proof}

\section{Binomial DGLM Simulation Results}
\label{sec:simextra}
In the main article, we did not show the binomial DGLM forecasting results for the bounded, zero-inflated Poisson simulation because it performed much worse than all other methods and therefore skewed the graphics.  For completeness we display here in Figures \ref{fig:ZIPcal} and \ref{fig:ZIPsharp} the corresponding calibration and sharpness plots with binomial DGLM included.

\begin{figure}
    \centering
    \includegraphics[width=.6\linewidth,height=\textheight,keepaspectratio]{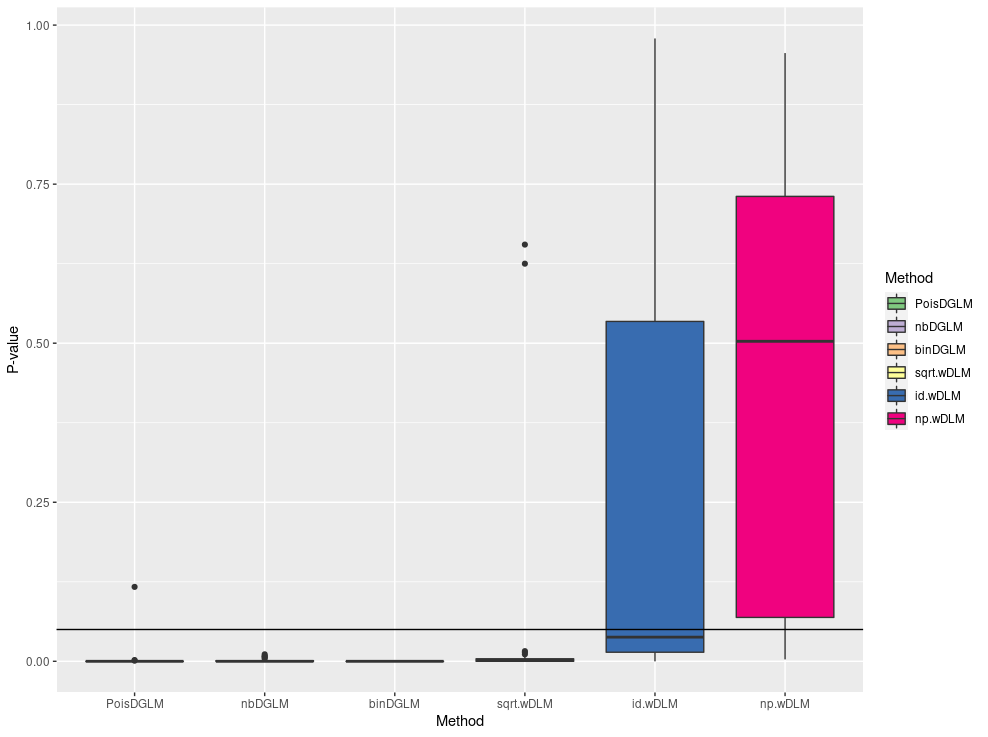}
    \caption{Box plot of p-values measuring calibration (larger is better) across simulations, with a line at p=0.05}
    \label{fig:ZIPcal}
\end{figure}

\begin{figure}
    \centering
    \includegraphics[width=.6\linewidth,height=\textheight,keepaspectratio]{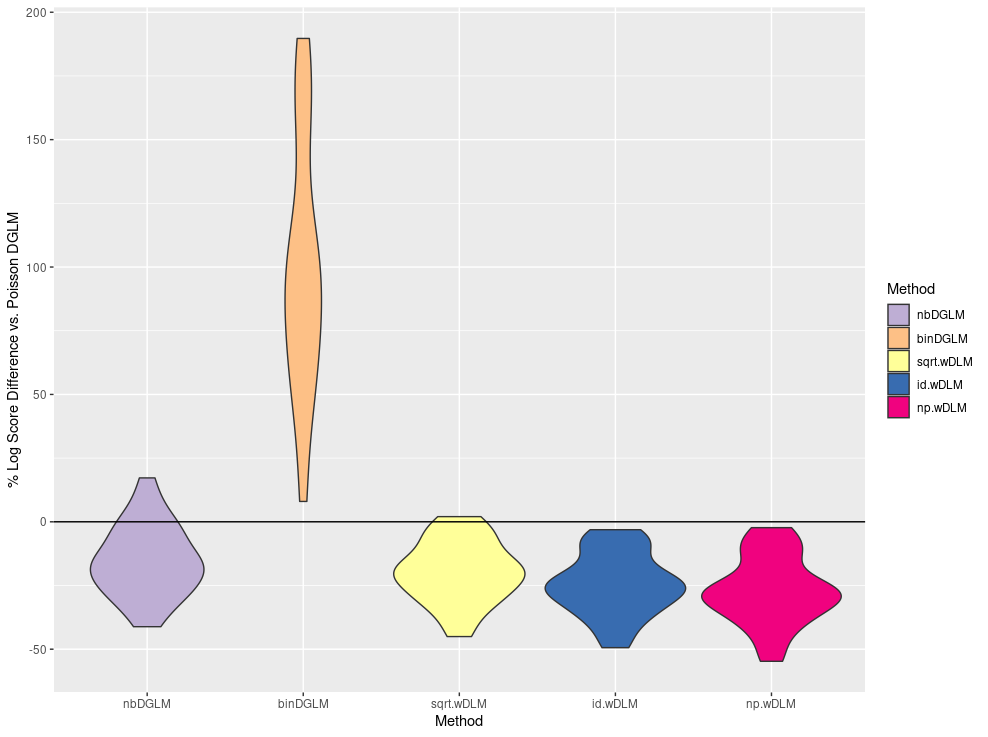}
    \caption{Violin plot of percent difference in log score compared to baseline across simulations (negative difference implies improved forecasting relative to Poisson DGLM)}
    \label{fig:ZIPsharp}
\end{figure}

\section{Application Details}
\label{sec:appdeets}

\subsection{Dataset Cleaning}
The Cincinnati dataset contains all fire department data, but our focus in this application is only on the overdose responses. To extract the appropriate observations, we followed the data dictionary in selecting heroin overdose responses as those with rows whose \verb|INCIDENT_TYPE_ID| was one of "HEROIF", "HEROIN-COMBINED", "HEROINF - FIRE ONLY", "HEROINF-FIRE ONLY", "HERON F", or "HERONF".  For the ``other'' overdose responses, we took all events which correspond to EMS protocol 23 (Overdose/Poisoning), i.e. those rows whose \verb|INCIDENT_TYPE_ID| started with the number 23.

Once the appropriate responses were selected, the only cleaning necessary was to count the number of occurrences per day and properly format dates. We did so by binning the observations into days based on the \verb|CREATE_TIME_INCIDENT| and assigning zeroes on any days where no incidents were reported.

\subsection{Linear Growth Model Details}
The warpDLM latent specification for the application is an example of Seemingly Unrelated Time Series Equations (SUTSE), where the inter-series dependence is defined through the covariance matrix of the model errors. Given the linear growth model for each series as defined in equations \eqref{eq:SUTSE_observation}--\eqref{eq:SUTSE_slope}, the full bivariate model can be rewritten in the general DLM formulation. Specifically, with $\bm \theta_t = (\mu_{1,t}, \mu_{2,t}, \beta_{1,t}, \beta_{2,t})'$, the state evolution equation  is
\begin{equation}
\label{eq:SUTSE_state}
\bm \theta_t
=
\begin{bmatrix}
1 & 0 & 1 & 0\\
0 & 1 & 0 & 1\\
0 & 0 & 1& 0\\
0 & 0 & 0 & 1
\end{bmatrix}
\bm \theta_{t-1}
+ \bm w_t
\quad 
\bm w_t \sim N_4\left(\bm 0, \bm W = 
\left[\begin{array}{@{}c|c@{}}
\normalfont\Large \bm W_\mu
  & \bigzero \\
\hline
  \bigzero &
\normalfont\Large \bm W_\beta
\end{array}\right]\right)
\end{equation}
and the observation equation is
\begin{equation}
\label{eq:SUTSE_obs}
\bm z_t
=
\begin{bmatrix}
1 & 0 & 0 & 0\\
0 & 1 & 0 & 0
\end{bmatrix}
\bm \theta_t
+ \bm v_t
\quad 
\bm v_t \sim N_2\left(\bm 0, \bm V\right)
\end{equation}
for the latent $\bm z_t = (z_{1,t}, z_{2,t})$.

\subsection{Seasonal Model Specification}
In addition to the linear growth model presented above, we also implemented a SUTSE model with a seasonal term. There are two main ways to include seasonality within DLMs: dummy variables or Fourier-form/trigonometric seasonality. 
 For a time series with period $s$, the dummy approach requires $s-1$ state variables, whereas the trigonometric approach only requires $\lfloor(s/2)\rfloor$, or even less, since the number of trigonometric terms can be truncated to achieve a smoother and more parsimonious seasonal trend. The overdose data appears to have a yearly seasonality and we have daily data, so our period is $s=365$. For such a long period, Fourier-form seasonality is most reasonable.
 
 We used the same linear growth specification as described in Section \ref{subsec:realdata}, but added the first 5 harmonics from a trigonometric seasonal component of period 365.
 A linear growth DLM with Fourier-form seasonality inherits the same state evolution dynamics as \eqref{eq:SUTSE_level}-\eqref{eq:SUTSE_slope} for the local level and slope, but the observation equation adds a seasonal term $\gamma_{i,t}$ for each series $i \in {1,2}$. Thus, equation \eqref{eq:SUTSE_observation} becomes
 \begin{equation}
 z_{i,t} = \mu_{i,t} + \gamma_{i,t} + v_{i,t} \;.
 \end{equation}
We write the seasonal term $\gamma_{i,t}$ dynamics as follows:
\begin{align}
 \gamma_{i,t} &= \sum_{j=1}^{5} \gamma_{i, j,t} \\
\gamma_{i,j,t} &= \gamma_{i,j,t-1} \cos{\lambda_j} + \gamma^*_{i,j,t-1} \sin{\lambda_j} + w_{i,j,t-1}^\gamma \\
\gamma^*_{i,j,t} &= -\gamma_{i,j,t-1} \sin{\lambda_j} + \gamma^*_{i,j,t-1} \cos{\lambda_j} + w_{i,j,t-1}^{\gamma^*} 
\end{align}
where $j=1,\ldots,5$ and $\lambda_j=2\pi j / 365$. Inter-series dependence between the seasonal components can come from the error terms:  $\bm w_{j,t-1}^\gamma, \bm w_{j,t-1}^{\gamma^*} \overset{iid}{\sim}  N_2\left(\bm 0, \bm W_\gamma \right)$.
 For more details on Fourier-form seasonality and how it can be represented in traditional DLM form, consult Section 3.2.3 of \citesupp{petrisDLM} or Section 6.1 of \citesupp{HelskeKFAS}. 
 
The above model results in a state space of dimension 24: $\bm \theta_t = (\mu_{i,t}, \beta_{i,t}, \beta_{i,t}, \gamma_{i,j,t}, \gamma^*_{i,j,t})'$ for $i \in \{1,2\}$ and $j=1,\ldots,5$. We set the variance of the seasonal component $\bm W_\gamma$ to be 0, so that the periodic piece of the model was static. Using this model specification, we performed the same analysis, running first the offline Gibbs sampler followed by the particle filter. The seasonal component does not appear to have much effect, as the forecasting performance was much the same as the linear growth model, even at longer horizons. This can be seen in Figure \ref{fig:seasonal_fc}, where we compare the performance of the two model specifications on the heroin time series.  To do so, we compute the percent difference in log score between the seasonal model and the linear growth model across different forecast horizons. Interestingly, the linear growth model actually seems to produce slightly better forecasts on average at longer horizons, although the difference is very minimal.
\begin{figure}
    \centering
    \includegraphics[width=.8\linewidth, keepaspectratio]{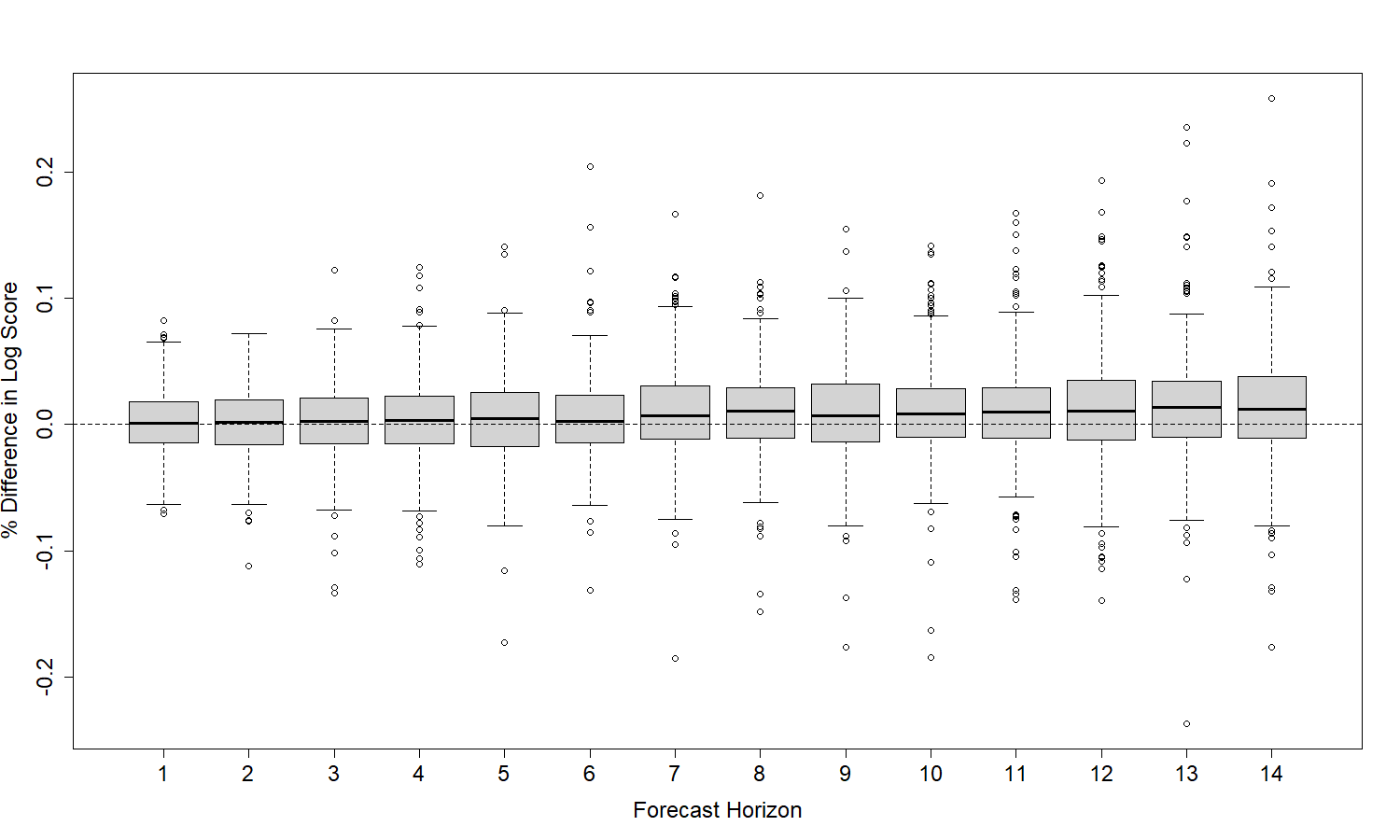}
    \caption{Percent difference in log score of seasonal model forecasts for heroin OD time series compared to linear growth model predictions, plotted across forecast horizons. The positive values for longer horizons suggest that the linear growth model outperforms the seasonal model for these longer-term forecasting distributions.}
    \label{fig:seasonal_fc}
\end{figure}

\subsection{Multi-Step Forecasting Results}
In the main text, we primarily focus on the 1-step-ahead forecasts.  Here, we investigate longer forecast horizons, which may be more relevant for planners and decision-makers in certain cases. We computed forecasts ranging from one day ahead to two weeks ahead (14 steps ahead with our daily data). Calibration for multi-step-ahead forecasts was usually less satisfactory, with the multi-step-ahead forecast distributions often overdispersed relative to the observed data. This is not surprising, given that the uncertainty of the forecast distribution accumulates with each additional forecast step. Figure \ref{fig:multistep} shows the boxplot of log scores across forecast horizons from the linear growth warpDLM model for both time series. It can be seen that the log scores steadily increase with the forecast horizon.
\begin{figure}
    \centering
    \includegraphics[width=\linewidth, keepaspectratio]{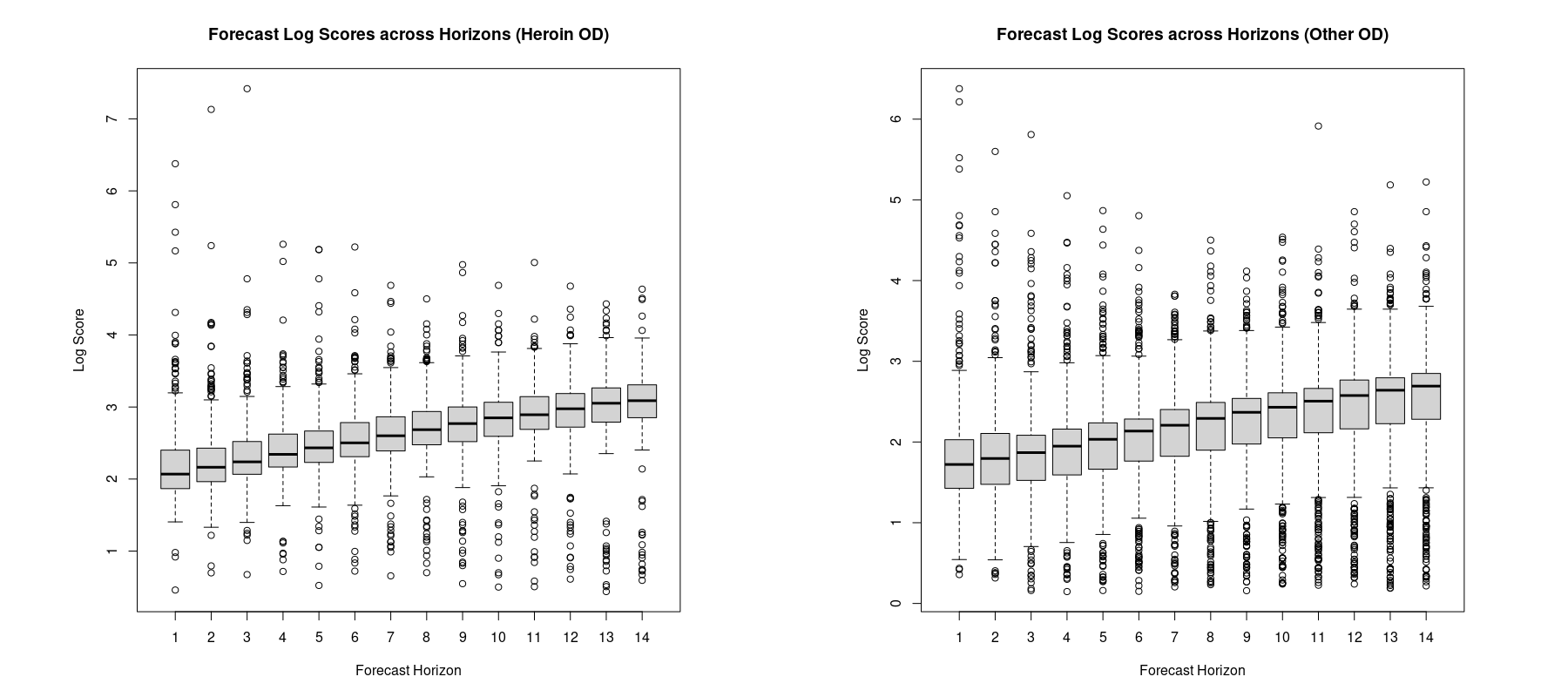}
    \caption{Boxplot of log scores from linear growth warpDLM model across different forecast horizons}
    \label{fig:multistep}
\end{figure}
%To improve these longer horizon forecasts, it may be necessary to try even more complex seasonal models than that presented above, perhaps by making the seasonal component stochastic. Alternatively, it may just be better to aggregate the time series at a higher temporal level and model from there if the focus is solely on several-week-ahead forecasts.
\section{Particle Filtering in Higher Dimensions}
\label{sec:particlesim}
In Section \ref{subsec:realdata}, we combined offline inference using the Gibbs sampler with online inference via the particle filter. We demonstrated that the particle filter results in  sensible smoothing and calibrated forecasts for a real-data application with $n=2$ and $T=1127$. However, we also wanted to explore how the particle filter performed in higher-dimensional settings, with larger $n$ and longer $T$.
To do so, we simulated a dataset with $n=5$ and $T=4000$. For the data-generating process, we used the multivariate Poisson-scaled Beta (MPSB) class developed in \citesupp{Aktekin2018}, as this provides a relatively simple method of simulating multivariate count time series data. Briefly, this class is described by the following hierarchical model:
\begin{align}
(Y_{jt} | \lambda_j, \theta_t) &\sim \mbox{Pois}(\lambda_j\theta_t) \quad j = 1,\ldots, n \; ; \; t=1,\ldots,T \\
\theta_t &= \frac{\theta_{t-1}}{\gamma} \epsilon_t \\
(\epsilon_t | \bm Y_{1:t-1}, \lambda_1, \ldots, \lambda_n) &\sim \mbox{Beta}(\gamma \alpha_{t-1}, (1-\gamma) \alpha_{t-1}) \\
\alpha_t &= \gamma\alpha_{t-1} + \sum_{j=1}^n Y_{jt} 
\end{align}
where $\alpha_{t-1} > 0$ , $0 < \gamma < 1$ and $D_{t-1} = \{D_{t-2}, Y_{1,t-1},\ldots,Y_{n,t-1}\}$.
The model is finalized by setting the prior for $\theta_0$, which is of the form $\theta_0 \sim \mbox{Gamma}(\alpha_0, \beta_0)$. To simulate from this model, it is also necessary to fix $\gamma$ and $\lambda_j$ for $j=1,\ldots,n$.
For our simulation, we generally followed the setup from Section 4.1 of \citesupp{Aktekin2018}, setting the series-specific factors $\lambda_j$ to $(2, 2.5, 3, 3.5, 4)$ and drawing the initial state from $\mbox{Gamma}(\alpha_0=100, \beta_0=100)$. However, we modified the common discount parameter $\gamma$, as the value of $0.3$ used in the \citesupp{Aktekin2018} simulation resulted in time series collapsing to zero for longer time scales. To prevent this collapse, we increased its value considerably, in the end fixing $\gamma=.93$. The resulting time series are plotted in Figure \ref{fig:multivariate_sim}.
\begin{figure}
    \centering
    \includegraphics[width=\linewidth, keepaspectratio]{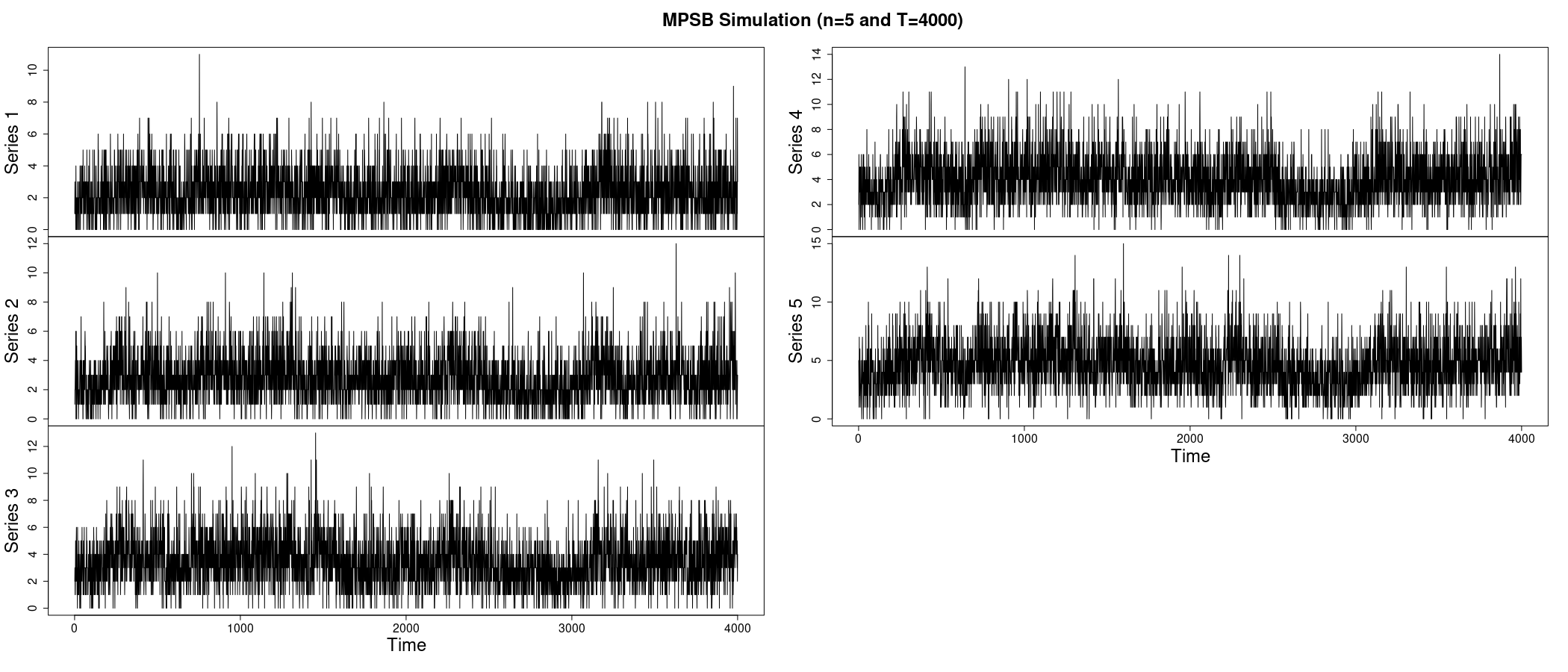}
    \caption{Time series plot of multivariate data simulated from the MPSB class}
    \label{fig:multivariate_sim}
\end{figure}
We modeled this multivariate time series with the same linear growth model as in the real-data application. Similarly, we split the dataset into offline and online portions, running the Gibbs sampler for the first 1000 data points, followed by the particle filter for the remaining 3000 time points. Our primary objective was to ensure the particle filter is able to accurately track the series in higher dimensions. We see that this goal is achieved by visualizing the ``push-forward" of our particles: for each particle $\theta_t^{(s)}$ from our filter (not to be confused with the $\theta_t$ of the MPSB model), we draw from $N(F_t\theta_t^{(s)}, V)$ and apply the inverse transformation and rounding to arrive at a value on the scale of our data. This is similar to the ``smoothing predictive distribution" procedure discussed in Section \ref{subsec:realdata}, only now we are looking at the ``filtering predictive distribution". As with the smoothing, we can take pointwise medians of our draws from this predictive filtering distribution; these are depicted in Figure \ref{fig:multisim_filtered}.
\begin{figure}
    \centering
    \includegraphics[width=\linewidth, keepaspectratio]{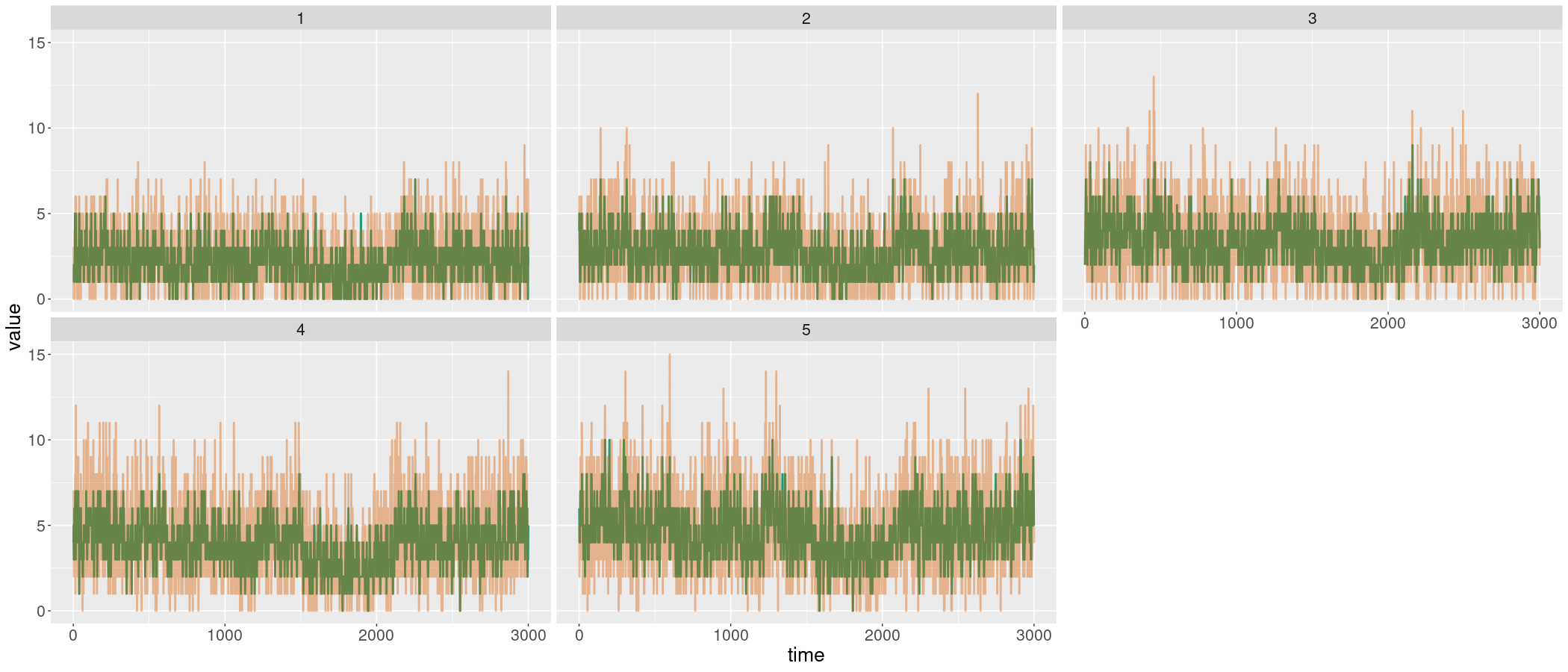}
    \caption{Multivariate online simulated dataset (light orange color) overlaid with the median filtered predictive distribution (dark green color)}
    \label{fig:multisim_filtered}
\end{figure}
The particle filter is tracking each series well. However, it is worth noting that the ESS in this setting is on average significantly lower than that of our real-data application: the ratio of average ESS to number of particles decreased from 58\% to around 24\%. This suggests that larger dimensions may require adoption of further SMC techniques like mutation, as was discussed in Section \ref{subsec:particlefilter}.
As a final note, the computation time for each time point increases with the number of series. This is due both to sampling from the transition density (involving repeated simulation from an $n$-variate multivariate truncated normal) and computation of the weights (evaluating an $n$-variate multivariate normal CDF). These steps must be done for each particle, although such computations could be parallelized for speed improvements. For the $n=5$ example with $S=5000$ particles, each time point is still executed in under 30 seconds, so our particle filter would remain applicable for minute-by-minute streaming data of this dimension. %This is especially true if the filter were deployed in an environment where computations could be parallelized across many cores.

\bibliographystylesupp{apalike}
\bibliographysupp{refs.bib}

\end{appendices}

\end{document}

\newif\ifblinded
\blindedtrue

\title{\Large Supplement to \\
``Warped Dynamic Linear Models for Time Series of Counts"}

\ifblinded
\author{\vspace{-20mm}}
\date{}
\else

 \author{Brian King\thanks{
    The authors gratefully acknowledge \textit{please remember to list all relevant funding sources in the unblinded version}}\hspace{.2cm}
    and 
    Daniel R. Kowal \thanks{Dobelman Family Assistant Professor, Department of Statistics, Rice University, Houston, TX 77251-1892 (\href{mailto:Daniel.Kowal@rice.edu}{Daniel.Kowal@rice.edu}).}}

\fi

\maketitle 
\appendix
\noindent
This supplementary document contains the following:
\begin{itemize}
    \item Section~\ref{sec-proofs}: Proposition and proof of SLCT-N closure under linear transformations as well as proofs of main results in the paper
    \item Section \ref{sec:gibbs}: Additional discussion of Gibbs sampling for the warpDLM
    \item Section \ref{sec:simextra}: Results from simulations that include binomial DGLM
    \item Section \ref{sec:appdeets}: Further details on application dataset and model
\end{itemize}

\section{Further Theory and Proofs}
\label{sec-proofs} 
As discussed in Section~\ref{subsec:selectionbasics}, the selection normal inherits closure under linear transformations, a result which is used to show that the warpDLM state predictive distribution stays within the selection normal family.  We formalize and prove this in the below proposition.
\begin{proposition}
\label{prop:linearity}
    Suppose $\bm \theta \sim \mbox{SLCT-N}_{n, p}(\bm \mu_z, \bm \mu_\theta, \bm \Sigma_z, \bm \Sigma_\theta, \bm \Sigma_{z\theta},  \mathcal{C})$ and let $\bm \theta^* \stackrel{d}{=} \bm A \bm\theta + \bm a$ where $\bm A$ is a fixed $q \times p$ matrix and $\bm a \sim N_q(\bm \mu_a, \bm \Sigma_a)$ is independent of $\bm \theta$. 
    
    Then $\bm \theta^* \sim \mbox{SLCT-N}_{n, q}(\bm \mu_z, \bm \mu_{\theta^*}, \bm \Sigma_z, \bm \Sigma_{\theta^*}, \bm \Sigma_{z\theta^*},  \mathcal{C})$ where $\bm \mu_{\theta^*} =\bm A\bm \mu_\theta + \bm \mu_a$, $\bm \Sigma_{\theta^*}= \bm A \bm \Sigma_\theta \bm A' +\bm \Sigma_a$, and $\bm \Sigma_{z\theta^*} = \bm \Sigma_{z\theta} \bm A'$.
\end{proposition}
\begin{proof}
    First consider the distribution of $\bm A \bm \theta$. Since the moments of the joint distribution $(\bm z, \bm \theta)$ are given by assumption, it follows that $(\bm z, \bm A \bm \theta)$ is jointly Gaussian with moments available by straightforward calculation, so $\bm A \bm \theta \sim \mbox{SLCT-N}_{n, q}(\bm \mu_z, \bm A\bm \mu_\theta, \bm \Sigma_z, \bm A \bm \Sigma_{\theta}\bm A', \bm \Sigma_{z\theta} \bm A',  \mathcal{C})$. 
    
    As noted in \cite{ArellanoValle2006}, selection normal distributions have an moment generating function (MGF).  For $[\bm \theta | \bm z \in \mathcal{C}] \sim \mbox{SLCT-N}_{n, p}(\bm \mu_z, \bm \mu_\theta, \bm \Sigma_z, \bm \Sigma_\theta, \bm \Sigma_{z\theta},  \mathcal{C})$, the MGF is:
    \begin{equation*}
    \label{mgf-slct-n}
    M_{[\bm \theta | \bm z \in \mathcal{C}]}(\bm s) = \exp\left(\bm s' \bm \mu_\theta + \frac{1}{2}\bm s' \bm\Sigma_\theta\bm s\right) \frac{\bar\Phi_n(\mathcal{C}; \bm\Sigma_{z\theta}\bm s + \bm \mu_z, \bm \Sigma_z)}{\bar\Phi_n(\mathcal{C}; \bm \mu_z, \bm \Sigma_z)}.
    \end{equation*}
    
    Using the selection normal MGF and noting independence between $\bm A \bm\theta$ and $\bm a$, it follows that $M_{\bm \theta^*}(\bm s) = M_{ \bm A \bm \theta+\bm a}(\bm s)= M_{ \bm A \bm \theta}(\bm s) M_{\bm a}(\bm s)$ where $ M_{ \bm A \bm \theta}$ is given by inserting the appropriate parameters into the equation above and $M_{\bm a}(\bm s) = \exp(\bm s' \bm \mu_a + \frac{1}{2}\bm s' \bm \Sigma_a \bm s)$. The result of this product is the MGF of the stated SLCT-N distribution.
\end{proof}

\noindent
We now state the proofs of all results in the main article, with the exception of Theorem \ref{normal-conjugacy}, whose result follows directly from the definition of a selection normal distribution and the warpDLM model setup. The concurrent work of \cite{kowal2021conjugate} also explores Theorem \ref{normal-conjugacy} and Lemma \ref{lem:conjugacy}, but with focus on non-dynamic linear regression for discrete data. 

\begin{proof}[Proof (Lemma \ref{lem:conjugacy})]
    The SLCT-N prior is equivalently defined by $[\bm \theta | \bm z_0 \in \mathcal{C}_0]$ for $(\bm z_0', \bm \theta')'$ jointly Gaussian with moments given in the prior. The posterior is constructed similarly: $[\bm \theta | \bm y] \stackrel{d}{=} [\bm \theta | \bm z_0 \in \mathcal{C}_0, \bm z \in g(\mathcal{A}_{\bm y})] \stackrel{d}{=} [\bm \theta | \bm z \in \mathcal{C}_0 \times g(\mathcal{A}_{\bm y})]$ where $\bm z_1 = (\bm z_0', \bm z')'$. It remains to identify the moments of $(\bm z_1', \bm \theta')' = (\bm z_0', \bm z', \bm \theta')'$. For each individual block of $\bm z_0$, $\bm z$, and  $\bm \theta$ and the pairs ($\bm z_0, \bm \theta$) and ($\bm z, \bm \theta$), the moments are provided by either the prior or the posterior in Theorem \ref{normal-conjugacy}. Finally, we have cross-covariance $\mbox{Cov}(\bm z_0, \bm z) = \mbox{Cov}(\bm z_0, \bm F \bm \theta + \bm v) =  \bm \Sigma_{z_0\theta}\bm F'$. % since  $\bm z_0$ is independent of $\bm \epsilon$.
\end{proof}

\begin{comment}
\begin{proof}[Proof (Lemma 3)]
We have from DLM properties (\cite{WestHarrisonDLM} Ch 16) that
\begin{equation*}
\begin{pmatrix} z_1 \\ \theta_1 \end{pmatrix} 
\sim N \left[
\begin{pmatrix} F_1a_1 \\ a_1 \end{pmatrix},
\begin{pmatrix} F_1 R_1 F_1' + V_1 & F_1 R_1\\ (F_1 R_1)' & R_1 \end{pmatrix}
\right]
\end{equation*}  Recognizing that
\[
(\bm \theta_1|y_1) = (\bm \theta_1|z_1 \in g(A_{y_1}))
\]
and leveraging the result from \cite{ArellanoValle2006}, we have our conclusion.
\end{proof}
\end{comment}

\begin{proof}[Proof (Theorem \ref{theorem:filtering})]
From the model setup, we have $\bm \theta_t =\bm  G_t \bm \theta_{t-1} + \bm w_t$ and $\bm \theta_{t-1} | \bm y_{1:t-1}$ follows a SLCT-N distribution as given.  Thus by Proposition \ref{prop:linearity}, we have our result that $\bm \theta_{t} | \bm y_{1:t-1}$ follows a SLCT-N with the given parameters.  

Then, with the state predictive distribution $\bm \theta_{t} | \bm y_{1:t-1}$ as our prior combined with the warpDLM likelihood \eqref{likelihood-time}, we can directly apply Lemma \ref{lem:conjugacy} to get the second result. 
\end{proof}

\begin{proof}[Proof (Theorem \ref{theorem:smooth})]
We have that $p(\bm \theta_{1:T}| \bm y_{1:T}) \propto p(\bm \theta_{1:T})p(\bm y_{1:T} | \bm \theta_{1:T})$.  By the model specification, we know $\bm \theta_{1:T} \sim N_{pT}(\bm \mu_\theta, \bm \Sigma_\theta)$ as defined above.  Furthermore, recall that given the latent states, the observations are conditionally independent.  Hence we can represent the likelihood as $p(\bm y_{1:T} | \bm \theta_{1:T}) = \prod_{t=1}^T \Bar{\Phi}_n(g(\mathcal{A}_{y_t}); \bm F_t\bm \theta_t, \bm V_t)$ which can be rewritten as $\Bar{\Phi}_{nT}(C; \boldsymbol{\mathfrak{F}}\bm \theta_{1:T}, \boldsymbol{\mathfrak{V}})$.  Therefore, we have
$$
p(\bm \theta_{1:T}| \bm y_{1:T}) \propto \phi_{pT}(\bm \theta_{1:T} ; \bm \mu_\theta, \bm \Sigma_\theta) \Bar{\Phi}_{nT}(C; \boldsymbol{\mathfrak{F}}\bm \theta_{1:T}, \boldsymbol{\mathfrak{V}})
$$
which is the kernel of a SLCT-N distribution with parameters given in \eqref{eq:smoothing}.
\end{proof}

\begin{proof}[Proof (Corollary \ref{corollary:marginallike})]
The marginal likelihood is simply the normalizing constant of the joint smoothing distribution, and its form is given by the denominator in \eqref{density-slct-n}. 
\end{proof}

\begin{proof}[Proof (Corollary \ref{corollary:pfresults})]
By Bayes' Theorem, we have
\begin{equation}
\label{pf-bayes}
p(\bm \theta_{t}| \bm \theta_{t-1}, \bm y_{t}) = \frac{p(\bm y_{t}| \bm \theta_{t})p(\bm \theta_{t}| \bm \theta_{t-1})}{p(\bm y_{t}| \bm \theta_{t-1})} \;.    
\end{equation}
Using the DLM and warpDLM properties, we know that
$$
p(\bm y_{t}| \bm \theta_{t}) = \Bar{\Phi}_n(g(\mathcal{A}_{y_t}); \bm F_t \bm \theta_t, \bm V_t) \quad \mbox{and} \quad p(\bm \theta_{t}| \bm \theta_{t-1}) = \phi_p(\bm G_t \bm \theta_{t-1}, \bm W_t) \;.
$$

Putting this together, recognize that the numerator of \eqref{pf-bayes} forms the kernel of a selection normal distribution \eqref{density-slct-n} with parameters given in \eqref{eq:importance}.  Furthermore, the marginal distribution $p(\bm y_{t}| \bm \theta_{t-1})$ is that kernel's normalizing constant, leading us to the result of \eqref{eq:weight}.
\end{proof}

\section{Gibbs Sampling}
\label{sec:gibbs}
From a Bayesian perspective, one benefit of latent data modeling is that a straightforward Gibbs sampler often results. Although we have the ability to directly sample from the posterior for the warpDLM, there are still advantages to using the Gibbs sampling approach.  First, as previously discussed, drawing from a multivariate truncated normal can become computationally intensive and somewhat error-prone in high dimensional situations.  The default Gibbs sampler does not suffer from this problem (computing time does scale with data size, but not sharply as with direct sampling).  Second, a Gibbs sampler allows us to incorporate model parameter priors and thus better quantify parameter uncertainty.  

The sampler for the warpDLM framework is similar to that developed in \cite{Kowal2020a}.  The Bayesian dynamic probit sampler \citep{albert1993bayes,Fasano2021} can be regarded as a special case.  In particular, if our data $y_t$ were binary, we could let our transformation $g$ be the identity while defining our rounding operator such that $\mathcal{A}_{y_t=0}=(-\infty,0)$ and $\mathcal{A}_{y_t=1}=(0, \infty)$, thus recovering the probit Gibbs sampling algorithm.

More explicitly, let $\bm \psi$ represent the non-state model parameters (e.g., the observation and evolution variances).  Our MCMC sampler has three basic steps, shown in Algorithm \ref{algo:gibbs}.

\begin{algorithm}
\label{algo:gibbs}
\begin{enumerate}
    \item \textbf{Sample the latent data}: draw $[\bm z_t|\mathcal{D}, \theta_{1:T}, \psi]$ from $N_n(\bm F_t \bm \theta_t, \bm V_t)$ truncated to $\bm z_t \in g(\mathcal{A}_{y_t})$ for $t=1,\ldots,T$
    \item \textbf{Sample the states}: draw $[\bm \theta_{1:T}|\bm z_{1:T}, \bm \psi]$ (i.e. from the typical DLM smoothing distribution)
    \item \textbf{Sample the parameters}: draw $[\bm \psi| \theta_{1:T}, z_{1:T}^*]$ (from the appropriate full conditionals depending on priors)
\end{enumerate}
 \caption{Gibbs Sampler for WarpDLM}
\end{algorithm}

In each pass, one can also easily draw from the posterior predictive distribution for $\bm z_t$ as well as the forecast distribution. Samples from the corresponding distributions for $\bm y_t$ are then given by computing $h\left(g^{-1}(\bm z_t)\right)$.

In step 1, we are sampling from many low-variate truncated normal distributions; hence the Gibbs sampling approach doesn't suffer from the same dimensionality problems as the direct sampling method.  In step 2, we can rely on already-developed Kalman filter-based simulation techniques, such as FFBS \citep{fruhwirth1994data,carter1994gibbs}, implemented in the R package \verb|dlm| \citep{petrisDLMPackage}, or the simulation smoother of \cite{durbin2002simple} found in the \verb|KFAS| package \citep{HelskeKFAS}.  In practice, the simulation smoother has better computing performance compared to FFBS.

In a fully Bayesian framework we would put priors on the variances, and sample from the appropriate full conditionals in step 3.  For univariate models, we put a Uniform$(0,A)$ prior with $A$ large on all standard deviations, leading to relatively simple Gibbs updates.  In the multivariate context, we place independent inverse-Wishart priors on variance matrices as in \cite{petrisDLM}.

As briefly mentioned in the main body of the article, we can also utilize the selection normal results of Section \ref{subsec:warpeddlmtheory} to construct a blocked Gibbs sampler.  In particular, instead of sampling the latent $\bm z_t$, we can just sample the states directly from the warpDLM joint smoothing distribution \eqref{eq:smoothing}, unconditional on $\bm z_t$.  By marginalizing over the latent data, there is some possibility to improve mixing within the Gibbs sampler, although each draw from this alternate Gibbs sampler may take longer.

\section{Additional Simulation Results}
\label{sec:simextra}
In the main article, we did not show the binomial DGLM forecasting results for the bounded, zero-inflated Poisson simulation because it performed much worse than all other methods and therefore skewed the graphics.  For completeness we display here in Figures \ref{fig:ZIPcal} and \ref{fig:ZIPsharp} the corresponding calibration and sharpness plots with binomial DGLM included.

\begin{figure}
    \centering
    \includegraphics[width=.6\linewidth,height=\textheight,keepaspectratio]{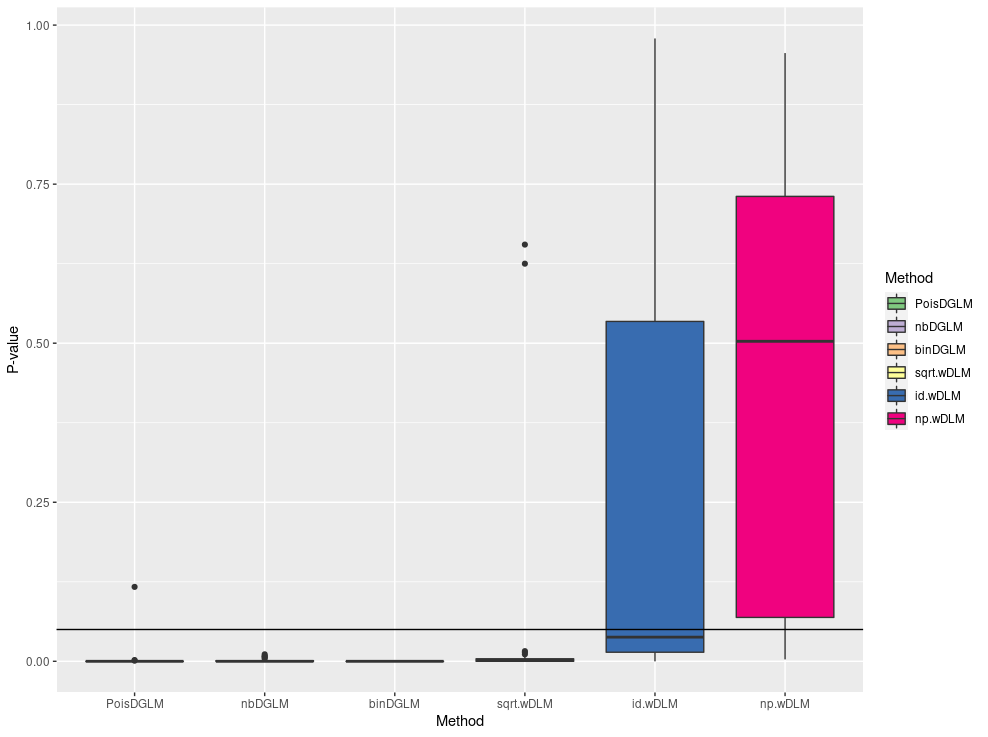}
    \caption{Box plot of p-values measuring calibration (larger is better) across simulations, with a line at p=0.05}
    \label{fig:ZIPcal}
\end{figure}

\begin{figure}
    \centering
    \includegraphics[width=.6\linewidth,height=\textheight,keepaspectratio]{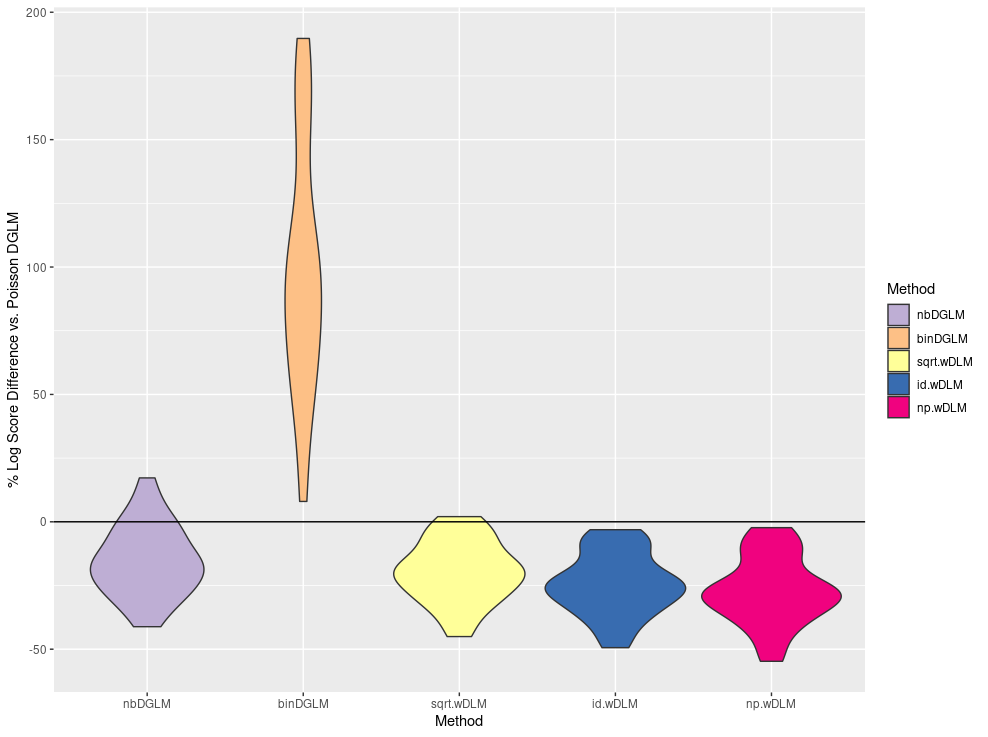}
    \caption{Violin plot of percent difference in log score compared to baseline across simulations (negative difference implies improved forecasting relative to Poisson DGLM)}
    \label{fig:ZIPsharp}
\end{figure}

\section{Application Details}
\label{sec:appdeets}

\subsection{Dataset Cleaning}
The Cincinnati dataset contains all fire department data, but our focus in this application is only on the overdose responses. To extract the appropriate observations, we followed the data dictionary in selecting heroin overdose responses as those with rows whose \verb|INCIDENT_TYPE_ID| was one of "HEROIF", "HEROIN-COMBINED", "HEROINF - FIRE ONLY", "HEROINF-FIRE ONLY", "HERON F", or "HERONF".  For the ``other'' overdose responses, we took all events which correspond to EMS protocol 23 (Overdose/Poisoning), i.e. those rows whose \verb|INCIDENT_TYPE_ID| started with the number 23.

Once the appropriate responses were selected, the only cleaning necessary was to count the number of occurrences per day and properly format dates. We did so by binning the observations into days based on the \verb|CREATE_TIME_INCIDENT| and assigning zeroes on any days where no incidents were reported.

\subsection{Model Details}
The warpDLM latent specification for the application is an example of Seemingly Unrelated Time Series Equations (SUTSE), where the inter-series dependence is defined through the covariance matrix of the model errors. Given the linear growth model for each series as defined in equations \eqref{eq:SUTSE_observation}--\eqref{eq:SUTSE_slope}, the full bivariate model can be rewritten in the general DLM formulation. Specifically, with $\bm \theta_t = (\mu_{1,t}, \mu_{2,t}, \beta_{1,t}, \beta_{2,t})'$, the state evolution equation  is
\begin{equation}
\label{eq:SUTSE_state}
\bm \theta_t
=
\begin{bmatrix}
1 & 0 & 1 & 0\\
0 & 1 & 0 & 1\\
0 & 0 & 1& 0\\
0 & 0 & 0 & 1
\end{bmatrix}
\bm \theta_{t-1}
+ \bm w_t
\quad 
\bm w_t \sim N_4\left(\bm 0, \bm W = 
\left[\begin{array}{@{}c|c@{}}
\normalfont\Large \bm W_\mu
  & \bigzero \\
\hline
  \bigzero &
\normalfont\Large \bm W_\beta
\end{array}\right]\right)
\end{equation}
and the observation equation is
\begin{equation}
\label{eq:SUTSE_obs}
\bm z_t
=
\begin{bmatrix}
1 & 0 & 0 & 0\\
0 & 1 & 0 & 0
\end{bmatrix}
\bm \theta_t
+ \bm v_t
\quad 
\bm v_t \sim N_2\left(\bm 0, \bm V\right)
\end{equation}
for the latent $\bm z_t = (z_{1,t}, z_{2,t})$.

\bibliographystyle{apalike}
\bibliography{refs.bib}
\makeatletter\@input{xx.tex}\makeatother